\documentclass[journal]{IEEEtran}

\usepackage{amsbsy}
\usepackage{algorithm}
\usepackage{algpseudocode}
\usepackage{amsmath,amsthm,amssymb,amsfonts}
\usepackage[english]{babel}
\usepackage{bm}
\usepackage[font=footnotesize,justification=justified,format=plain]{caption}
\usepackage{cite}
\usepackage{hhline}
\usepackage{cleveref}
\usepackage{enumerate}
\usepackage{float}
\usepackage{graphicx}
\usepackage{import}
\usepackage{multirow}
\usepackage{slashbox}
\usepackage{soul}
\usepackage[labelformat=simple,font=footnotesize,justification=justified,format=plain]{subcaption}
\usepackage{xcolor}
\newtheorem{theorem}{Theorem}
\theoremstyle{remark}
\newtheorem{remark}[theorem]{Remark}

\newcommand{\xzero}{\bm{0}}

\newcommand{\xlambda}{\boldsymbol{\lambda}}
\newcommand{\xTheta}{\boldsymbol{\Theta}}
\newcommand{\xzeta}{\boldsymbol{\zeta}}

\newcommand{\xa}{\mathbf{a}}
\newcommand{\xA}{\mathbf{A}}
\newcommand{\xb}{\mathbf{b}}
\newcommand{\xB}{\mathbf{B}}
\newcommand{\xc}{\mathbf{c}}
\newcommand{\xC}{\mathbf{C}}
\newcommand{\xd}{\mathbf{d}}
\newcommand{\xD}{\mathbf{D}}
\newcommand{\xe}{\mathbf{e}}
\newcommand{\xf}{\mathbf{f}}
\newcommand{\xF}{\mathbf{F}}
\newcommand{\xG}{\mathbf{G}}

\newcommand{\xH}{\mathbf{H}}

\newcommand{\xI}{\mathbf{I}}

\newcommand{\xJ}{\mathbf{J}}
\newcommand{\xn}{\mathbf{n}}

\newcommand{\xo}{\mathbf{o}}
\newcommand{\xO}{\mathbf{O}}

\newcommand{\xQ}{\mathbf{Q}}

\newcommand{\xR}{\mathbf{R}}
\newcommand{\xs}{\mathbf{s}}
\newcommand{\xt}{\mathbf{t}}
\newcommand{\xu}{\mathbf{u}}
\newcommand{\xU}{\mathbf{U}}
\newcommand{\xv}{\mathbf{v}}
\newcommand{\xV}{\mathbf{V}}

\newcommand{\xx}{\mathbf{x}}
\newcommand{\xX}{\mathbf{X}}
\newcommand{\xy}{\mathbf{y}}
\newcommand{\xY}{\mathbf{Y}}
\newcommand{\xz}{\mathbf{z}}

\newcommand{\xxri}{_{ri}}
\newcommand{\xxrin}{_{rin}}
\newcommand{\xxrk}{_{rk}}
\newcommand{\xxklr}{_{klr}}
\newcommand{\xxkr}{_{kr}}
\newcommand{\xkRi}{_{k.i}} 
\newcommand{\xkRk}{_{k.k}} 
\newcommand{\xkRj}{_{k.j}} 
\newcommand{\xRsumk}{\xR_{.k}}
\newcommand{\Rklk}{\xR\xkl^k}
\newcommand{\Rklj}{\xR\xkl^j}
\newcommand{\Rjmk}{\xR\xjm^k}
\newcommand{\xin}{_{i,n}}
\newcommand{\xjm}{_{j,m}}

\newcommand{\xxki}{_{ki}}
\newcommand{\xklrin}{_{klrin}}
\newcommand{\xklrkl}{_{klrkl}}
\newcommand{\xklrjm}{_{klrjm}}

\newcommand{\xkm}{_{k,m}}
\newcommand{\xxkj}{_{kj}}
\newcommand{\xxkk}{_{kk}}
\newcommand{\xkl}{_{k,l}}
\newcommand{\xnxk}{_{n_k}}

\newcommand{\yth}{^{th}}
\newcommand{\yin}{^{i,n}}
\newcommand{\yjm}{^{j,m}}
\newcommand{\ykl}{^{k,l}}
\newcommand{\ysp}{^{s+1}}
\newcommand{\ybps}{^{\backprime s}}
\newcommand{\ybpsp}{^{\backprime s+1}}

\newcommand{\tm}{\text{-}}

\newcommand{\tR}{\text{R}}
\newcommand{\tT}{\text{T}}
\newcommand{\tmax}{\text{max}}
\newcommand{\tre}{\text{re}}
\newcommand{\tSINR}{\text{SINR}}
\newcommand{\tSNR}{\text{SNR}}
\newcommand{\ttr}{\text{tr}}
\newcommand{\tr}{\text{rank}}
\newcommand{\tvec}{\text{vec}}

\algnewcommand{\WhileDoEnd}[3]{
  \State \algorithmicwhile\ #1\ \algorithmicdo\ #2\ \algorithmicend\ #3}

\algnewcommand{\IfThenEnd}[3]{
  \State \algorithmicif\ #1\ \algorithmicthen\ #2\ \algorithmicend\ #3}

  \algnewcommand{\IfThenEndwIndent}[3]{
  \State ~~~ \algorithmicif\ #1\ \algorithmicthen\ #2\ \algorithmicend\ #3}

\newcounter{problemcounter}
\newcommand{\problemcounter}[1]{\refstepcounter{problemcounter} \text{P}\theproblemcounter} 

\newtheorem{lemma}{Lemma}
\newtheorem{proposition}{Proposition}

\newcommand{\mIIInVIIIrX}   {\text{K} 3\text{-}\text{M} 3 \text{-} \text{N} 8 \text{-} \text{R} 10}
\newcommand{\mIVnVIIIrIX}   {\text{K} 3\text{-}\text{M} 4 \text{-} \text{N} 8 \text{-} \text{R} 9}

\newcommand{\mVInVIIIrXIV}  {\text{K} 3\text{-}\text{M} 6 \text{-} \text{N} 8 \text{-} \text{R} 14}

\newcommand{\mVIIInVIIIrIII} {\text{K} 3\text{-}\text{M} 8 \text{-} \text{N} 8 \text{-} \text{R} 3}
\newcommand{\mXnXrIII}      {\text{K} 3\text{-}\text{M} 10 \text{-} \text{N} 10 \text{-} \text{R} 3}

\newcommand{\mXnVIIIrIII}   {\text{K} 3\text{-}\text{M} 10 \text{-} \text{N} 8 \text{-} \text{R} 3} 
\newcommand{\mXnVIIIrIV}    {\text{K} 3\text{-}\text{M} 10 \text{-} \text{N} 8 \text{-} \text{R} 4}

\newcommand{\mXVnVIIIrX}    {\text{K} 3\text{-}\text{M} 15 \text{-} \text{N} 8 \text{-} \text{R} 10}

\begin{document}

\setstcolor{red}
\captionsetup[figure]{labelformat=simple,labelsep=period,name={Fig.}}
\renewcommand{\thesubfigure}{(\alph{subfigure})}

\title{Distributed Multi-Stream Beamforming in MIMO Multi-Relay Interference Networks}

\author{Cenk M. Yetis and Ronald Y. Chang
\thanks{This work was supported in part by the Ministry of Science and Technology, Taiwan, under Grant MOST 106-2628-E-001-001-MY3.

The authors are with the Research Center for Information Technology Innovation, Academia Sinica, Taipei 115, Taiwan (e-mail: cenkmyetis@ieee.org; rchang@citi.sinica.edu.tw).
}}
\maketitle

\begin{abstract}
In this paper, multi-stream transmission in interference networks aided by multiple amplify-and-forward (AF) relays in the presence of direct links is considered. The objective is to minimize the sum power of transmitters and relays by beamforming optimization under the stream \mbox{signal-to-interference-plus-noise-ratio} (SINR) constraints. For
transmit  beamforming optimization, the problem is a \mbox{well-known} \mbox{non-convex} quadratically constrained quadratic program (QCQP) that is NP-hard to solve. After \mbox{semi-definite} relaxation (SDR), the problem can be optimally solved via alternating direction method of multipliers (ADMM) algorithm for  distributed implementation. Analytical and extensive numerical analyses demonstrate that the proposed ADMM solution converges to the optimal centralized solution. The convergence rate, computational complexity, and message exchange load of the proposed algorithm outperforms the existing solutions.  Furthermore, by SINR approximation at the relay side, distributed joint transmit and relay beamforming optimization is also proposed that further improves the total power saving at the cost of increased complexity.

 \end{abstract}

\IEEEpeerreviewmaketitle

\section{Introduction}
The advent of future wireless networks bearing new components in large numbers including eNodeBs, cloud servers, relays, smart grids, massive multiple-input multiple-output (MIMO), and big data nodes, and the recent advances both in software and hardware architectures surging the applicability of parallel and distributed computations \cite{719} have brought innovative solutions and paradigm shifts in recent years\cite[Chap. 10]{722}\cite{383,477,723}. However, the distributed solutions based on conventional dual decomposition and other methods lack the effectiveness on the numerical stability, fast convergence rates \cite{695}, and scalability to high dimensional problems \cite{721} compared to alternating direction method of multipliers (ADMM) which combines the strengths of dual decomposition and augmented Lagrangian methods \cite{724}.

ADMM studies in relay \cite{682,676,714} and point-to-point \cite{684,698,713,710,709,721} networks are limited.
In \cite{682}, an improved version of ADMM is proposed for power minimization under SINR constraints in multi-cluster relay networks with single antenna nodes.
In \cite{676}, max-min SINR optimization is studied for decode-and-forward (DF) relay networks with single antenna transmitters and receivers under the constraints of total transmitter power, total relay power, and total number of relays. In \cite{714}, a distributed transmit power control algorithm via ADMM is proposed for a single relay aided network. ADMM in point-to-point networks is slightly more investigated in the literature. In \cite{721}, cloud radio access networks (C-RAN) with single antenna receivers are optimized for power minimization under SINR constraints. In \cite{684}, power minimization with the worst-case SINR constraints due to the channel state information (CSI) errors in multi-cell coordinated multiple-input single-output (MISO) networks is solved via \mbox{semi-definite} relaxation (SDR). In \cite{698}, power minimization problem with rate constraints in wireless sensors networks is studied. In \cite{713}, beamforming design for power minimization under SINR constraints is proposed in MISO downlink systems by solving second order cone problems (SOCP). In \cite{710}, max-min flow rate optimization problem is considered for software defined radio access networks (SD-RAN) with single antenna nodes under wired and wireless link constraints. In \cite{709}, a distributed power control algorithm is proposed for the utility maximization without SINR constraints in interference networks with single antenna nodes.

The three prominent features of future high performance wireless networks are multi-stream transmissions, \mbox{multi-antenna} nodes, and the line-of-sights between transmitters and receivers. The existence of direct links between transmitters and receivers is an immediate outcome of deploying a large number of intermediate nodes in the network to close the distances between transmitters and receivers. In C-RANs, SD-RANs, and wireless relay networks, the wireless intermediate nodes are called radio access units \cite{721}, base stations \cite{710}, and relays, respectively. Wireless relay networks can be regarded as the wireless communications parts of C-RANs and SD-RANs, which embody wired communications parts in their architectures as well. Although the three mentioned features are critical in practical systems, many studies on relay networks are based on  communications settings with lesser features due to the difficulties that arise from the coexistence of all three features\cite{692}. Among many objective functions \cite[Chap. 8]{214}, the mean square error (MSE) minimization problems are more tractable in relay networks with the three features. Nevertheless, the studies on MSE are still limited\cite{725}, and to the best of our knowledge, the MSE problem in a relay network with all the three features  is studied only in \cite{716}.

The limited attributes of aforementioned researches on relay networks can diminish their applications  in future practical systems. In this paper, we consider amplify-and-forward (AF) multi-relay interference networks with all the three mentioned attributes. The problem of interest in this network setting is to minimize the total power consumption of transmitters and relays with guaranteed quality of service (QoS) by distributed optimization of the transmit beamforming filters.
The QoS metric chosen in this paper is stream SINR, which is interconnected to the data rate and bit error rate (BER) expressions. The problem is a member of non-convex quadratically constrained quadratic programming (QCQP) problems \cite{728,729}, which is not directly amenable to distributed optimization due to the intricate stream SINR constraints. In particular, the SINR constraints are not in a linear form and also are not decoupled over streams for parallel and distributive implementation. We fit the problem into the ADMM framework, a potent tool for distributed optimization.

When designing a solution for the problem, the feasibility of problem, i.e., the  SINR constraints (targets) must be jointly supported, must be assured in the first stage before solving the main problem in the second stage. There are two approaches to assure the success of the first stage, namely, deriving the feasibility conditions \cite{217} and relaxing the initial conditions, e.g., searching for  feasible SINR targets\cite{206,477} and reducing the number of users \cite{711}. The derivation of the feasibility condition is challenging even in simpler networks. In \cite{654}, an approximate condition is derived for a multi-relay network with a single transmitter and a receiver. The feasibility search, on the other hand, can be as costly
as the solving the main problem in terms of the number of iterations and computational complexity per iteration.

In this paper, we propose to use random initializations of beamforming vectors to automatically determine feasible SINR targets with a high probability. All cases where the convergence is slow, fluctuant, and infeasible, i.e., the SINR targets are infeasible, make up less than $2\%$ of the simulations presented in this paper. The elimination of these
mentioned cases  is beneficial in two important applications: (1) Accurate and extensive cross-analyses of crucial network parameters, and benchmarking with competitive schemes over these varying parameters can be executed in short times. (2) By weighting the auto assigned  SINR targets with scalar variables, the feasibility search problem can be reduced to a simpler linear search problem as demonstrated in Section \ref{subsec:SINRtargets}. The \mbox{cross-analyses} are accurate since no approximations are needed in contrast to the approximately derived feasibility conditions \cite{654}.

Relay beamforming design in the existence of direct links has been a long standing open problem due to the challenge in the expression of
SINR metric in terms of relay filters as detailed in Section \ref{sec:Joint}. In this paper, an SINR reformulation is proposed that gives good approximation when the direct channels
and the effective channels between the transmitters and receivers are independent. The proposed distributed joint transmit and relay beamforming assures improved total power saving
than that of only the distributed transmit beamforming optimization. However, since each relay serves all streams in the network, the complexity substantially increases.

The main contributions of this paper are summarized as follows:
\begin{itemize}
\item A generic network model with multiple multi-antenna nodes at all sides, i.e., the transmitter, relay, and receiver sides, to carry out multi-stream transmissions in the presence of direct links is transformed into a compact matrix system model.
\item The proposed distributed ADMM algorithm achieves the optimal centralized solutions in the given generic network model. In addition, in terms of convergence rate, computational complexity, and message exchange load performance metrics, the proposed solution surpasses other optimal distributed algorithms.
\item By eliminating the feasibility of problem stage automatically, an extensive evaluation of the effects of crucial network parameters on the system performance metrics are attained that reveals new insights in this paper. \item Analytical and numerical results demonstrate the lower computational complexity and message exchange load, higher convergence rate, i.e., lesser number of iterations, and
finally the convergence of proposed distributed algorithm.
\item SINR approximation at the relay side is proposed to implement distributed joint transmit and relay beamforming optimization that further improves the total power saving at the cost of increased complexity.
\end{itemize}

The closest to our contribution in this paper is given in [11], where ADMM is also applied. The major difference between \cite{714} and this paper is that a scalar power variable in a single relay network and a beamforming vector in a multi-relay aided network are optimized, respectively. Clearly, the extension in this paper is nontrivial. Further differences between \cite{714} and this paper are discussed in later sections.

ADMM framework is applicable to many problems under mild conditions. The methodology is different than the distributed algorithms that are  designed
for particular problems. In \cite{55}, a distributed algorithm is proposed for interference alignment in signal space. To achieve the alignment in a distributed manner, each user minimizes the interference covariance matrix. Such distributed solutions
that are designed for particular areas  are in contrary to  the distributed ADMM solutions that have broad application areas. In fact, ADMM solution for interference alignment
is already exploited  in \cite{727}.

The rest of the paper is organized as follows. The multi-stream transmission capable multi-relay interference network is introduced in Section \ref{sec:SystemModel}. In Section \ref{sec:ProblemFormulation}, the transmit beamforming design problem for power minimization under stream SINR constraints is formulated. The distributed ADMM solution and benchmark distributed solutions are presented in Section \ref{sec:MultiStream}. The attributes of distributed solutions including convergence, computational complexity, and message exchange load are studied in Section \ref{sec:Attributes}. Distributed joint transmit and relay beamforming filter optimization via SINR approximation at the relay side is provided in Section \ref{sec:Joint}. The numerical results and discussions are presented in Section \ref{sec:NumericalResults}, and finally, the paper is concluded with the summary of main results in Section \ref{sec:Conclusion}.

\begin{figure*}[!t]
 \vspace{-.3cm}
\centering
\includegraphics[height=8.5cm, width=10.5cm] {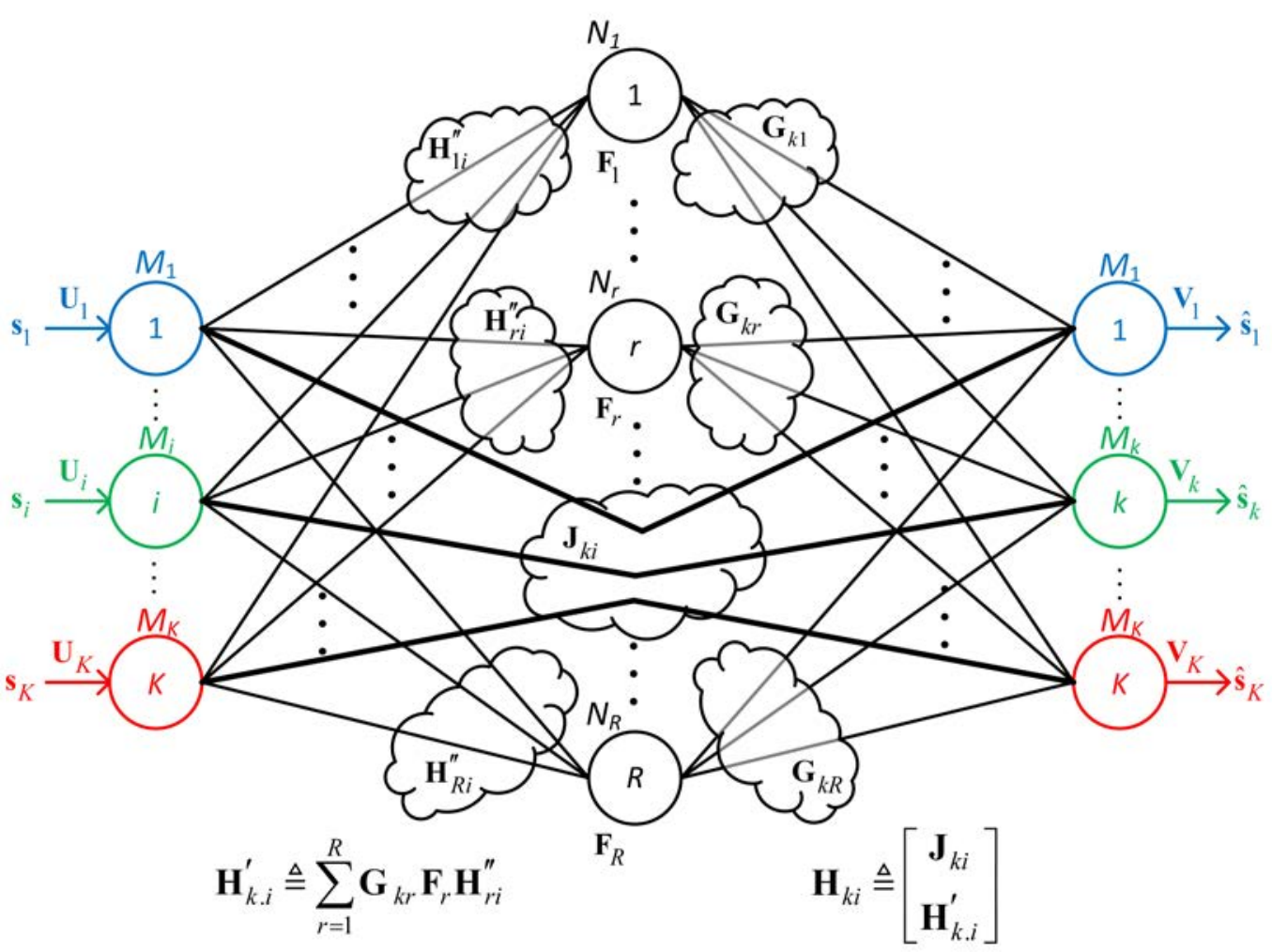}
 \caption{Multi-stream per user transmission in a $K$-user MIMO multi-relay interference network with direct links. Direct links are indicated by the bold lines.}
 \label{Fig:SystemModel}
 \vspace{-.3cm}
\end{figure*}

\section{System Model} \label{sec:SystemModel}

Consider a $K$-user two-hop MIMO multi-relay interference network aided by $R$ relays. The $i\yth$ source (transmitter) and the $i\yth$ destination (receiver) each has $M_i$ antennas while the $r\yth$ relay has $N_r$ antennas as shown in Fig. \ref{Fig:SystemModel}. Each transmitter communicates with its corresponding receiver with the aid of all relays. Without loss of generality, $i\yth$ transmitter and receiver pair can be called $i\yth$ user. We assume that all relay nodes work in half-duplex mode. Thus the communication between the users is completed in two time slots and there are non-negligible direct links between all transmitters and receivers. In the first time slot, the $i\yth$ transmitter transmits the ${M_i\times1}$ signal vector \mbox{$\xx_i=\xU_i\xs_i$}, where \mbox{$\xU_i\in\mathbb{C}^{M_i\times d_i}$} and \mbox{$\xs_i\in\mathbb{C}^{d_i}$} are the transmit beamforming matrix and symbol vector with \mbox{$E\big(\xs_i\xs_i^H\big)=\xI_{d_i}$} and ${E\big(\xs_i\xs_j^H\big)=\xzero}$ for \mbox{$j\neq i$}, respectively. Here, $d_i$ is the number of streams of the $i\yth$ user, i.e., the number of independent data streams to be transmitted between the $i\yth$ transmitter and receiver pair. We assume $d_i\leq M_i$ for sufficient degrees of freedom in signal detection. The transmitted signal from the user has a power constraint
\begin{equation}\label{eqn:180321_PC1}
p_i=E(||\xx_i||^2)=\ttr\big(\xU_i\xU_i^H\big)\leq p_i^\tmax,
\end{equation}
where $p_i^\tmax$ is the maximum power of the $i\yth$ transmitter. The received signal at the $r\yth$ relay and $k\yth$ receiver in the first time slot is given by
\begin{equation}\xy_r=\sum_{i=1}^{K}\xH^{\prime\prime}\xxri\xx_i + \xn_r \text{ and }\xy_k(1)=\sum_{i=1}^{K}\xJ\xxki\xx_i + \xn_k(1),
\end{equation}
 respectively, where $\xH^{\prime\prime}\xxri$ is the channel from the $i\yth$ transmitter to the $r\yth$ relay, and $\xJ\xxki$ is the channel between the $i\yth$ transmitter and the $k\yth$ receiver. $\xn_r$ and $\xn_k(1)$ are the complex additive white Gaussian noise (AWGN) at the $r\yth$ relay and at the $k\yth$ receiver in the first time slot with zero mean, and with the covariances ${E\big(\xn_r\xn_r^H\big)=\sigma_r^2\xI_{N_r}}$ and ${E\big(\xn_k(1)\xn_k^H(1)\big)=\sigma_k^2(1)\xI_{M_k}}$, respectively. In the second time slot, the received signal $\xy_r$ is precoded at the $r\yth$ relay by the ${N_r\times N_r}$ relay filter $\xF_r$, ${\xx_r=\xF_r\xy_r}$. The $r\yth$ relay transmit power is ${p_r=E(||\xx_r||^2)=\ttr\big(\xF_r\xR_{y_r}\xF_r^H\big)\leq p_r^\tmax}$, where ${\xR_{y_r}=E\big(\xy_r\xy_r^H\big)}$ and $p_r^\tmax$ is the covariance matrix of the received signal at and the maximum power of the $r\yth$ relay, respectively. The relay transmit power can be explicitly written as
 \begin{equation}\label{eqn:180321_PC2}
 p_r=\ttr\Big(\!\!\sum_{i=1}^{K}\!\!\xU_i^H\xH_{ri}^{\prime\prime H}\xF_r^H\xF_r\xH^{\prime\prime}_{ri}\xU_i\Big) + \sigma_r^2\ttr\Big(\xF_r\xF_r^H\Big).
 \end{equation}
 The received signal at the $k\yth$ receiver in the second time slot is given by
 \begin{align}\label{eqn:180407_1549}
\!\!  \xy_k(2)\!\!=&\!\!\sum_{r=1}^R\xG_{kr}\xx_r+\xn_k(2)=\sum_{r=1}^R\xG_{kr}\xF_r\xH^{\prime\prime}_{rk}\xx_k\nonumber\\
  \!\!&\!\!+\!\sum_{r=1}^R\xG_{kr}\xF_r\!\!\!\!\!\!\sum_{\substack{j=1,j\neq k}}^{K}\!\!\!\xH^{\prime\prime}_{rj}\xx_j\!+\!\!\sum_{r=1}^R\xG_{kr}\xF_r\xn_r+\xn_k(2),
  \end{align}
 where $\xG_{kr}$ is the channel between the $r\yth$ relay and the $k\yth$ receiver.

Define the effective channel from transmitter $i$ to receiver $k$ through all relays as
\begin{equation}\label{eqn:effectivechannel}
 \xH^\prime\xkRi\triangleq\sum_{r=1}^R\xG_{kr}\xF_r\xH_{ri}^{\prime\prime}.
\end{equation}
 Then, $\xy_k(2)$ in \eqref{eqn:180407_1549} is rewritten as $$\xy_k(2)=\xH^\prime\xkRk\xx_k+\!\sum_{\substack{j=1,j\neq k}}^{K}\xH^\prime\xkRj\xx_j+\sum_{r=1}^R\xG_{kr}\xF_r\xn_r+\xn_k(2).$$ To obtain the SINR expression, the received signal at a receiver is written in terms of the desired signal, interference signal, and noise summands. Thus, define the aggregate channel matrix and noise vector as
\begin{equation}
\xH\xxki\triangleq\left[
       \begin{array}{c}
        \!\!\xJ\xxki\!\! \\
        \!\!\xH^\prime\xkRi\!\! \\
       \end{array}
      \right]\! \text{ and } \xn_k\triangleq\!\left[
  \begin{array}{c}
   \!\!\xn_k(1)\!\! \\
   \!\!\sum_{r=1}^R\xG_{kr}\xF_r\xn_r+\xn_k(2)\!\! \\
  \end{array}
 \right].
\end{equation}
The list of channel notations used in the paper is given in Table \ref{tab:NotationsList} for convenience. Then, the aggregate received signal at receiver $k$ can be written as
\begin{equation}
\xy_k\triangleq\left[
       \begin{array}{c}
        \xy_k(1) \\
        \xy_k(2) \\
       \end{array}
      \right]=\xH\xxkk\xU_k\xs_k+\sum_{\substack{j=1,j\neq k}}^{K}\xH\xxkj\xU_j\xs_j+\xn_k.
\end{equation}
For the sake of linear decoding complexity, the intra- and inter-user stream interferences are treated as noise in this paper. The receive filter is $\bar{\xV}_k\in\mathbb{C}^{2M_k\times d_k}$ since the receive filters for the time slots $1$ and $2$, $\xV_k(1)\in\mathbb{C}^{M_k}$ and $\xV_k(2)\in\mathbb{C}^{M_k}$, respectively, are stacked in this matrix, i.e., $\bar{\xV}_k\triangleq\left[
       \begin{array}{c}
        \xV_k(1) \\
        \xV_k(2) \\
       \end{array}
      \right]$. After applying the receive filter $\bar{\xV}_k$ to the received signal $\xy_k$, the SINR of the $l\yth$ stream of the $k\yth$ user is obtained as
\begin{equation}\label{eqn:180328_1610}
\tSINR\xkl=\frac{\tilde{\zeta}\xkl\ykl}{\sum_{(j,m)\neq(k,l)}\tilde{\zeta}\xkl\yjm+\sigma_{n\xkl}^2},
\end{equation}
where \mbox{$\tilde{\zeta}\xkl\yin\triangleq\big|\bar{\xv}\xkl^H\xH\xxki\xu\xin\big|^2$} and $\bar{\xv}\xkl\triangleq[\xv\xkl^T(1)\,\xv\xkl^T(2)]^T$ is the \mbox{$2M_k\times1$} receive beamforming vector for the $l\yth$ stream of the $k\yth$ user, i.e., $l\yth$ column of the receiver filter $\bar\xV_k$ (similarly $\xv\xkl(i)$, ${i=1,2}$ are the $l\yth$ columns of the receiver filters $\xV_k(i)$, ${i=1,2}$, respectively), \mbox{$\sigma_{n\xkl}^2=\bar{\xv}\xkl^H\xR\xnxk\bar{\xv}\xkl$} is the power of the aggregate noise after the receive beamforming 
\begin{equation}\label{eqn:180515_1330}
\xR\xnxk\!\!\triangleq\!\!\left(\!
\begin{array}{cc}
\!\!\sigma_k^2(1)\xI_{M_k}\!\! & \!\!\!\boldsymbol{0}\!\!\! \\
\!\!\boldsymbol{0}\!\! &\!\!\!\sum_{r=1}^R\!\sigma_r^2\xG_{kr}\xF_r\xF_r^H\xG_{kr}^H\!+\!\sigma_k^2(2)\xI_{M_k}\!\!\\
\end{array}
\!\right)\end{equation} is the covariance matrix of the aggregate noise, and the notation $(j,m)\neq(k,l)$ denotes that $j\neq k$ and/or $m\neq l$.

\section{Problem Formulation}\label{sec:ProblemFormulation}
In this paper, the total power of transmitters and relays under SINR per stream constraints is minimized via the distributed ADMM algorithm \cite{695} by optimizing all stream beamforming filters $\xu\xkl,\,\forall k,l$ in parallel at each stream, aka processor, in the network. From \eqref{eqn:180321_PC1} and \eqref{eqn:180321_PC2}, the total power, i.e., the sum of transmitter and relay powers, can be rewritten in terms of stream filters $\xu\xkl$ ($l\yth$ column of the transmitter filter $\xU_k$) as
\begin{equation}
\sum_{k=1}^{K}\sum_{l=1}^{d_k}\!\xu\xkl^H\xRsumk\xu\xkl,
\end{equation} where ${\xRsumk\triangleq\xI_{M_k}+\xRsumk^\prime, \xRsumk^\prime\triangleq\sum_{r=1}^R\xR\xxrk \text{, and}}$
$${\xR\xxrk\triangleq\xH\xxrk^{\prime\prime H}\xF_r^H\xF_r\xH\xxrk^{\prime\prime}.}$$

Similarly, SINR of a stream can be rewritten as
\begin{equation}\label{eqn:SINR}
 \tSINR\xkl=\frac{\xu\xkl^H\Rklk\xu\xkl}{\sum_{(j,m)\neq(k,l)}\xu\xjm^H\Rklj\xu\xjm+\sigma_{n\xkl}^2},
\end{equation}
where \mbox{$\xR\xkl^i\triangleq\xH\xxki^H\bar{\xv}\xkl\bar{\xv}\xkl^H\xH\xxki$}. The superscript index in $\xR\xkl^i$ indicates the $i\yth$ transmitter, since $\xH\xxki$ is the aggregate channel between the $i\yth$ transmitter and $k\yth$ receiver. The simplifications of $\xR$ matrices are listed in Table \ref{tab:SimplificationsList} for convenience.

The problem formulation is given as

\vspace{2cm}\problemcounter{}\label{op:0720_1620} \vspace{-.6cm}
\begin{align}
&\mkern-5mu \underset{\{\!\xu\xkl\!\}_{\forall k\in\mathcal{K},\forall l\in\mathcal{L}_k}}\min & & \mkern-15mu\sum_{k=1}^{K}\sum_{l=1}^{d_k}\xu\xkl^H\xRsumk\xu\xkl\nonumber\\
& \hspace{.7cm} \text{s.t.} & & \mkern-10mu\tSINR\xkl \geq \gamma\xkl,~\forall k\in\mathcal{K},~\forall l\in\mathcal{L}_k
\end{align}
where ${\mathcal{K}\triangleq\{1,2,\ldots,K\}}$ and ${\mathcal{L}_k\triangleq\{1,2,\ldots,d_k\}}$ is the set of users and set of streams of the $k\yth$ user, respectively.

\begin{table}
\begin{center}
\caption{LIST OF CHANNEL NOTATIONS.} \label{tab:NotationsList}\vspace{-.2cm}
\begin{tabular}{|l|l|}
 \hline
 $\xJ\xxki$ & Channel from the $i\yth$ transmitter to the $k\yth$ receiver\\
 $\xG_{kr}$ & Channel from the $r\yth$ relay to the $k\yth$ receiver \\
 $\xH_{ri}^{\prime\prime}$ & Channel from the $i\yth$ transmitter to the $r\yth$ relay \\
 $\xH^\prime\xkRi$ & Effective channel from the $i\yth$ transmitter to the $k\yth$ receiver \\
 & through all relays \\
 $\xH\xxki$ & Aggregate effective channel from the $i\yth$ transmitter to the \\
 & $k\yth$ receiver, i.e., stacked matrix of $\xJ\xxki$ and $\xH^\prime\xkRi$\\
 \hline
\end{tabular}
\end{center}
\vspace{-.5cm}
\end{table}
\noindent

\section{Multi-Stream Beamforming Under Stream SINR Constraints}\label{sec:MultiStream}

\begin{table}
\begin{center}
\caption{LIST OF SOME SIMPLIFICATIONS.} \label{tab:SimplificationsList} \vspace{-.2cm}
\begin{tabular}{|c|c|}
 \hline
$\xR\xxrk$&$\xH\xxrk^{\prime\prime H}\xF_r^H\xF_r\xH\xxrk^{\prime\prime}$\\
$\xRsumk^\prime$&$\sum_{r=1}^R\xH\xxrk^{\prime\prime H}\xF_r^H\xF_r\xH\xxrk^{\prime\prime}$ \\
$\xRsumk$&$\xI_{M_k}+\xRsumk^\prime$\\
$\xR\xkl^i$&$\xH\xxki^H\bar{\xv}\xkl\bar{\xv}\xkl^H\xH\xxki$\\
 \hline
\end{tabular}
\end{center}
\vspace{-.5cm}
\end{table}

P\ref{op:0720_1620} is a \mbox{non-convex} quadratically constrained quadratic programming (QCQP) problem \cite{728,729}. To obtain the $\xu\xkl$ filters, P\ref{op:0720_1620} can be equivalently rewritten by using the fact that \mbox{$\xu\xkl^H\xR\xkl^i\xu\xkl=\ttr(\xX\xkl\xR\xkl^i)$}, where \mbox{$\xX\xkl\triangleq\xu\xkl\xu\xkl^H$}, and by rewriting the SINR constraint as a summation inequality rather than a division inequality. P\ref{op:0720_1620} can be rewritten as

\vspace{.2cm}\problemcounter{}\label{op:0727_1557} \vspace{-.6cm}
\begin{subequations}
\begin{align}
&\mkern-5mu \underset{\{\xX\xkl\}_{\forall k,l}}\min & & \mkern-8mu\sum_{k=1}^{K}\sum_{l=1}^{d_k}\ttr(\xX\xkl\xRsumk)\\
& \hspace{.3cm}\text{s.t.} & & \mkern-15mu \frac{1}{\gamma\xkl}\ttr(\xX\xkl\Rklk)-\!\!\!\!\!\!\!\!\!\sum_{(j,m)\neq(k,l)}\!\!\!\!\!\!\!\!\!\ttr(\xX\xjm\Rklj)\!\geq\!\sigma_{n\xkl}^2,\nonumber\\
&\mkern-20mu & &\mkern120mu \forall k\in\mathcal{K},\,\forall l\in\mathcal{L}_k\label{op:0727_1557_1stConstraint}\\
&\mkern-20mu & &\mkern-15mu \xX\xkl\in\mathbb{S}_+^{M_k},~\forall k\in\mathcal{K},~\forall l\in\mathcal{L}_k\label{op:0727_1915_SPSDConstraint}\\
&\mkern-20mu & &\mkern-15mu \tr(\xX\xkl)=1,~\forall k\in\mathcal{K},~\forall l\in\mathcal{L}_k.\label{op:0727_1915_RankConstraint}
\end{align}
\end{subequations}
The transmit beamforming covariance matrix constraint \eqref{op:0727_1915_SPSDConstraint} imposes the convex constraint that $\xX\xkl$ matrix belongs to the cone of symmetric and positive semi-definite matrices of dimension $M_k$ (denoted by $\mathbb{S}_+^{M_k}$). Note that, since the covariance matrices in \eqref{op:0727_1557_1stConstraint}, i.e., $\xR\xkl^k$ and $\xR\xkl^j $, are Hermitian, \mbox{$\ttr(\xX\xkl\xR\xkl^i)\in\mathbb{R}$}, which means that the SINR inequality constraint \eqref{op:0727_1557_1stConstraint} is well defined. However, P\ref{op:0727_1557} is still non-convex due to the last constraint \eqref{op:0727_1915_RankConstraint}, hence SDR can be applied, i.e., the last constraint can be relaxed \cite{728,729}. Nonetheless, the resulting SDR of P\ref{op:0727_1557} still cannot be solved distributively. To obtain $\xX\xin$ at the stream ${(i,n)}$ via a parallel and distributed approach, both the objective function and constraints in P\ref{op:0727_1557} must be separable with respect to each stream. The objective function of P\ref{op:0727_1557} is separable; however, the reformulated SINR constraints \eqref{op:0727_1557_1stConstraint} are coupled. Moreover, since the SINR constraints are not linear, ADMM is not directly applicable to P\ref{op:0727_1557}.

\subsection{Proposed ADMM Algorithm}\label{subsec:ProposedADMM}
Before presenting the proposed ADMM algorithm, we briefly review the main steps of ADMM. ADMM can solve the following convex problem

P($\xx,\xy$) \vspace{-.6cm}
\begin{subequations}
\begin{align}\label{op:ADMMSample}
&\underset{\xx\in\mathbb{C}^m,\xy\in\mathbb{C}^n}\min & & \hspace{-1.2cm}f(\xx)+g(\xy)\\
& \hspace{.5cm} \text{s.t.} & & \hspace{-1.2cm} \xA\xx+\xB\xy=\xc\label{op:ADMMSampleConstraint}\\
& & &  \hspace{-1.2cm} \xx\in\mathcal{S}_1,\xy\in\mathcal{S}_2,
\end{align}
\end{subequations}
where $\xA\in\mathbb{C}^{k\times m}$, $\xB\in\mathbb{C}^{k\times n}$, $\xc\in\mathbb{C}^k$, the functions $f$ and $g$ are convex, $\mathcal{S}_1$ and $\mathcal{S}_2$ are nonempty convex sets. Then the augmented Lagrangian for P($\xx,\xy$) is given as
\begin{align}\label{op:ADMMSampleLagrangian}
\mathfrak{L}_{\rho}(\xx,\xy,\xz)=&f(\xx)+g(\xy)+\tre(\xz^H(\xA\xx+\xB\xy-\xc))\nonumber\\
&+\frac{\rho}{2}\|\xA\xx+\xB\xy-\xc\|_2^2,
\end{align}
where $\xz\in\mathbb{C}^k$ is the Lagrange multiplier of the constraint \eqref{op:ADMMSampleConstraint}, re(.) is the real part operator, and $\rho$ is again the Lagrangian dual update step size. P($\xx,\xy$) is solved by ADMM via three steps at each iteration $s$ as follows
\begin{subequations}
\begin{align}
\xx\ysp=&\arg\underset{\xx}\min~\mathfrak{L}_{\rho}(\xx,\xy^s,\xz^s)\label{op:ADMMSola}\\
\xy\ysp=&\arg\underset{\xy}\min~\mathfrak{L}_{\rho}(\xx\ysp,\xy,\xz^s)\label{op:ADMMSolb}\\
\xz\ysp=&\xz^s+\rho(\xA\xx\ysp+\xB\xy\ysp-\xc).\label{op:ADMMSolc}
\end{align}
\end{subequations}

In order to transform the reformulated SINR constraints to a linear form in P\ref{op:0727_1557}, we initially introduce an auxiliary variable for each summand in the SINR constraint \eqref{op:0727_1557_1stConstraint} as
\begin{align}
\zeta\xkl\triangleq&\frac{1}{\gamma\xkl}\ttr(\xX\xkl\Rklk)-\sigma_{n\xkl}^2 \text{ and }\nonumber\\
\zeta\xkl^\backprime\triangleq&-\!\!\!\!\sum_{(j,m)\neq(k,l)}\ttr(\xX\xjm\Rklj).
\end{align}
Then, since the SINR constraints are active at the optimal point,  i.e., the SINR constraints \eqref{op:0727_1557_1stConstraint} must hold with equality \cite{732}, the resulting SDR of P\ref{op:0727_1557} can be equivalently rewritten as

\vspace{.2cm}\problemcounter{}\label{op:0727_1915} \vspace{-.6cm}
\begin{subequations}
\begin{align}
&\mkern-5mu \underset{\{\xX\xkl,\zeta\xkl,\zeta\xkl^\backprime\}_{\forall k,l}}\min & & \mkern-10mu\sum_{k=1}^{K}\sum_{l=1}^{d_k}\ttr(\xX\xkl\xRsumk)\\
& \hspace{.6cm} \text{s.t.} & &\mkern-75mu\zeta\xkl+\zeta\xkl^\backprime=0,\forall k\in\mathcal{K},\forall l\in\mathcal{L}_k\label{op:0727_1915_SINR1stConstraint}\\
&\mkern-20mu & &\mkern-75mu\zeta\xkl\!=\!\frac{1}{\gamma\xkl}\ttr(\xX\xkl\Rklk)\!-\!\sigma_{n\xkl}^2,\forall k\!\in\!\mathcal{K},\forall l\!\in\!\mathcal{L}_k\label{op:0727_1915_SINR2ndConstraint}\\
&\mkern-20mu & &\mkern-75mu\zeta\xkl^\backprime\!=\!-\!\!\!\!\!\!\!\!\!\sum_{(j,m)\neq(k,l)}\!\!\!\!\!\!\!\!\!\ttr(\xX\xjm\Rklj),\forall k\!\in\!\mathcal{K},\forall l\!\in\!\mathcal{L}_k\label{op:0727_1915_SINR3rdConstraint}\\
&\mkern-20mu & & \mkern-75mu \text{The constraint }\eqref{op:0727_1915_SPSDConstraint}. \nonumber
\end{align}
\end{subequations}
Now, the SINR constraint \eqref{op:0727_1557_1stConstraint} is in a linear form via the constraints \eqref{op:0727_1915_SINR1stConstraint}, \eqref{op:0727_1915_SINR2ndConstraint}, and \eqref{op:0727_1915_SINR3rdConstraint}, and the coupling constraint \eqref{op:0727_1915_SINR1stConstraint} is a simple linear constraint that is viable for the ADMM algorithm.
The partial augmented Lagrangian for P\ref{op:0727_1915} can be written as
\begin{align}\label{eqn:1803011340}
&\mathfrak{L}_\rho\big(\{\xX\xkl,\zeta\xkl,\zeta\xkl^\backprime,\lambda\xkl\}_{\forall k,l}\big)=\!\sum_{k=1}^{K}\sum_{l=1}^{d_k}\!\Big(\ttr(\xX\xkl\xRsumk)\nonumber\\
&\hspace{2.5cm}+\lambda\xkl(\zeta\xkl\!+\!\zeta\xkl^\backprime)\!+\!\frac{\rho}{2}(\zeta\xkl\!+\!\zeta\xkl^\backprime)^2\Big),
\end{align}
where $\lambda\xkl$ is the Lagrange multiplier of the constraint \eqref{op:0727_1915_SINR1stConstraint}, and $\rho\in\mathbb{R}_{+}$ is a positive constant parameter for adjusting the convergence speed, i.e., Lagrangian dual update step size.

P($\xx,\xy$) and P\ref{op:0727_1915} have the following correspondences
\begin{align}\label{eqn:180329_0047}
\xx=&[x_{11} \ldots x_{kl} \ldots x_{Kd_K}]^T,~~\xy=[\zeta\xkl \, \zeta\xkl^\backprime]^T\nonumber\\
f(\xx)=&\sum_{k=1}^{K}\sum_{l=1}^{d_k}\ttr(\xX\xkl\xRsumk),~~g(\xy)=0\nonumber\\
\xA=&\xzero,~~ \xB=[1 \,1]^T,~~\xc=\xzero\nonumber\\
\mathcal{S}_1=&\Big\{\xx|\xX\xkl\in\mathbb{S}_+^M\Big\},\nonumber\\
\mathcal{S}_2=&\bigg\{\zeta\xkl\in\mathbb{R}|\zeta\xkl\!=\!\frac{1}{\gamma\xkl}\ttr(\xX\xkl\Rklk)\!-\!\sigma_{n\xkl}^2,\forall k,l,\nonumber\\
&\zeta\xkl^\backprime\in\mathbb{R}|\zeta\xkl^\backprime\!=\!-\!\!\!\!\!\!\!\!\!\sum_{(j,m)\neq(k,l)}\!\!\!\!\!\!\!\!\!\ttr(\xX\xjm\Rklj),\forall k,l\bigg\},
\end{align}
where $x_{kl}\triangleq f_x\left(\ttr(\xX\xkl\xRsumk)\right)$ and $f_x$ is a mapping function between $x_{kl}$ and $\xX\xkl$.

As mentioned earlier, in \cite{714}, total power minimization under SINR constraints problem is solved via a distributed power control algorithm in a simpler network. In contrast, in this work, the problem is solved via beamforming vectors in a generic network. Hence, the problem in this work is more challenging and also the results are more effective than \cite{714}, i.e., higher sum-SINRs can be achieved with lesser power consumptions. Another distinction between \cite{714} and this paper is the difference between the augmented Lagrangian functions, where both of the proposed solutions are fundamentally based on. As seen in \cite[P2]{714}, the auxiliary definitions \cite[Eqs. (3c), (3d)]{714} are augmented in the Lagrangian function \cite[Eq. (4)]{714} to obtain closed-form solutions, whereas in this work, the summation of auxiliary terms \eqref{op:0727_1915_SINR1stConstraint} are augmented as seen in \eqref{eqn:1803011340}. In \cite{714}, power control algorithm is proposed for a simpler relay network architecture. In other words, transmit power scalar variables $p\xkl$ are optimized in \cite{714} as opposed to the transmit beamforming filters $\xu\xkl$ in this work. The summation of auxiliary terms in \cite[Eq. (3b)]{714} cannot be augmented in the Lagrangian function. Otherwise, the $p\xkl$ variable disappears in the Lagrangian differentiation process, hence closed-form solutions cannot be obtained. On the other hand, the auxiliary definitions \eqref{op:0727_1915_SINR2ndConstraint} and \eqref{op:0727_1915_SINR3rdConstraint} cannot be augmented due to the trace operator in this work. Therefore, in this work, \eqref{op:0727_1915_SINR1stConstraint} is augmented in the Lagrangian function and the covariance matrix solutions are obtained via CVX. Resorting to CVX for the solutions of covariance matrices is a common practice in the literature \cite{684,682}. After obtaining the transmit covariance matrices $\xX\xkl$, obtaining the transmit vectors $\xu\xkl$ is a well-known process \cite{728}, i.e., if the obtained covariance matrix is not rank-one, then additional rank-one approximate solution methods can be applied including the Gaussian randomization method.
However, as observed by the numerical results in Section \ref{sec:NumericalResults}, the optimal $\xX\xkl$ matrices are always \mbox{rank-one}. This indicates that the proposed distributed optimal solution to P\ref{op:0727_1915} serves as a global optimal solution to P\ref{op:0720_1620}.

\begin{remark}
P\ref{op:0727_1915} is separable with respect to the streams, thus it requires solving $B$ problems of P${\big(\xX\xkl,\zeta\xkl,\zeta\xkl^\backprime\big)}$, where ${B\triangleq\sum_{i=1}^K d_i}$ is the total number of streams in the network. However, in order to apply the two-block ADMM algorithm which is better understood than $x$-block ADMM algorithms in terms convergence \cite{726}, where ${x\!>\!2}$, the variables ${\big\{\xX\xkl,\zeta\xkl,\zeta\xkl^\backprime\big\}_{\forall k\in\mathcal{K},\,\forall l\in\mathcal{L}_k}}$ can be divided into two groups $\{\xX\xkl\}$ and ${\big\{\zeta\xkl,\zeta\xkl^\backprime\big\}}$. Hence, P\ref{op:0727_1915} is separated into two simpler parts: P($\xX\xkl$) and P${\big(\zeta\xkl,\zeta\xkl^\backprime\big)}$. Therefore, ADMM can distributively solve P\ref{op:0727_1915} by solving $2B$ simpler subproblems in parallel.
\end{remark}
\subsubsection*{The Main Steps of ADMM} ADMM consists of sequential updates of primal variables ${\xX\xkl}$, $\zeta\xkl$, and $\zeta^\backprime\xkl$, and the dual variables ${\lambda\xkl,\mu\xkl,\text{and }\mu\xkl^\backprime}$ \cite{719} as presented below.

\paragraph{Update of $\xX\xkl$} \label{subsubsec:UpdateX} The smaller problem of $\xX\xkl$ is

\vspace{.2cm}P($\xX\xkl$) \vspace{-.4cm}
\begin{align}\label{eqn:180329_0113A}
&\hspace{1cm} \underset{\xX\xkl}\min & &\mkern-10mu \ttr(\xX\xkl\xRsumk)\nonumber\\
&\hspace{1.1cm} \text{s.t.} & &\mkern-10mu \text{The constraints } \eqref{op:0727_1915_SPSDConstraint}^\prime \text{ and } \eqref{op:0727_1915_SINR2ndConstraint}^\prime,
\end{align}
where the superscript $^\prime$ denotes the constraints to be considered without the symbol $\forall$, e.g., \eqref{op:0727_1915_SINR2ndConstraint}$^\prime$ denotes the constraint only for the stream $(k,l)$, not for all streams ${\forall k\in\mathcal{K}}$, $\forall l\!\in\!\mathcal{L}_k$ as seen in \eqref{op:0727_1915_SINR2ndConstraint}.
The above linear problem can be solved locally at each processor $(k,l)$ in parallel.

\paragraph{Updates of $\zeta\xkl \text{ and } \zeta\xkl^\backprime$} The smaller problem of $\zeta\xkl \text{ and } \zeta\xkl^\backprime$ is

\vspace{.2cm}P($\zeta\xkl,\zeta\xkl^\backprime$) \vspace{-.4cm}
\begin{align}
&\hspace{1cm} \underset{\zeta\xkl,\zeta\xkl^\backprime}\min & & \hspace{-.5cm} \lambda\xkl(\zeta\xkl\!+\!\zeta\xkl^\backprime)\!+\!\frac{\rho}{2}(\zeta\xkl\!+\!\zeta\xkl^\backprime)^2\nonumber\\
&\hspace{1.25cm}\text{s.t.}& &\hspace{-.5cm}\text{The constraints } \eqref{op:0727_1915_SINR2ndConstraint}^\prime \text{ and } \eqref{op:0727_1915_SINR3rdConstraint}^\prime.
\end{align}
Then according to the Karush-Kuhn-Tucker (KKT) condition, the updates of auxiliary variables are given as
\begin{subequations} \label{eqn:0818_1542}
\begin{align}
&\zeta\xkl\ysp=-\frac{\lambda\xkl^s+\mu\xkl^s}{\rho}-\zeta\xkl^{\backprime s} \text{ and }\label{eqn:0818_1542_zeta}\\
& \zeta\xkl\ybpsp=-\frac{\lambda\xkl^s+\mu\xkl\ybps}{\rho}-\zeta\xkl\ysp,\label{eqn:0818_1542_zetab}
\end{align}
\end{subequations}
where $\mu\xkl^s$ and $\mu\xkl\ybps$ are the Lagrangian multipliers of the constraints \eqref{op:0727_1915_SINR2ndConstraint} and \eqref{op:0727_1915_SINR3rdConstraint}, respectively.

\paragraph{Updates of $\lambda\xkl,\mu\xkl,\text{and }\mu\xkl^\backprime$} The updates of dual variables ${\lambda\xkl,\mu\xkl,\text{and }\mu\xkl^\backprime}$ are given as
\begin{subequations}\label{eqn:0818_1548}
\begin{align}
\!\!&\lambda\xkl\ysp\!=\!\lambda\xkl^s+\rho\!\left(\zeta\xkl\ysp+\zeta\xkl\ybpsp\right),\label{eqn:0818_1548_lambda}\\
\!\!&\mu\xkl\ysp\!=\!\mu\xkl^s\!+\!\rho_c\!\left(\zeta\xkl\ysp\!-\!\frac{1}{\gamma\xkl}\ttr\left(\xX\xkl\ysp\Rklk\right)\!+\!\sigma_{n\xkl}^2\right), \label{eqn:0818_1548_mu}\\
\!\!&\mu\xkl^{\backprime s+1}\!=\!\mu\xkl^{\backprime s}+\rho_c\!\left(\zeta\xkl^{\backprime s+1}+\!\!\!\!\!\sum_{(j,m)\neq(k,l)}\!\!\ttr\left(\xX\xjm\ysp\Rklj\right)\right),\label{eqn:0818_1548_mub}
\end{align}
\end{subequations}
where ${\rho_c\in\mathbb{R}_{+}}$ is the conventional Lagrangian dual update step size.

\subsection{Pseudocode}\label{subsec:Pseudocode}
The pseudocode of the proposed ADMM algorithm is given in Algorithm \ref{alg:0818_1524}, where Matlab scripting language is used. As seen in step \ref{step:Success} of the algorithm, when the absolute deviation from the SINR target of each stream is within $\Delta\xkl^\tmax$ accuracy, the algorithm is terminated with a success flag.

\begin{algorithm}[tb]\small
\caption{Proposed distributed multi-stream beamforming algorithm via ADMM for MIMO multi-relay interference networks.} \label{alg:0818_1524}
\begin{algorithmic}[1]
\Statex 1) Randomly initialize Lagrangian variables, $\{\zeta\xkl^s,\zeta\xkl^{\backprime s},\lambda\xkl^s,\mu\xkl^s,\mu\xkl^{\backprime s}\}_{\forall k,l}$. 2) Initialize the stream powers equally, $\ttr(\xu\xkl^s\xu\xkl^{Hs})=p_k^\tmax/d_k$, $\forall k,\forall l$, and relay powers, $p_r=p_r^\tmax, \forall r$. 3) Compute the stream SINRs, $\tSINR\xkl$, by randomly initializing the transmit and relay beamforming vectors, and assign the SINR target of a stream of user $k$ to the average SINR of user $k$, $\gamma_k=\big(\sum_{l=1}^{d_k}\tSINR\xkl\big)/d_k$.
\State $iflag=0$, $\%$ $iflag$: infeasibility flag
\State $s=0$
\While {$s\leq s^\tmax-1$}
\State Obtain $\xX\xkl\ysp$ from P($\xX\xkl$),$\forall k,\forall l$ \label{step:ADMM1st}
\State Obtain ($\zeta\xkl\ysp$,$\zeta\xkl\ybpsp$), $\lambda\xkl\ysp$, $\mu\xkl\ysp$, and $\mu\xkl\ybpsp,\forall k,\forall l$ from \eqref{eqn:0818_1542}, \eqref{eqn:0818_1548_lambda}, \eqref{eqn:0818_1548_mu}, and \eqref{eqn:0818_1548_mub}, respectively \label{step:ADMM2nd}
\State $\Delta\xkl\ysp=|\tSINR\xkl\ysp-\gamma_k|,\forall k,\forall l$
\IfThenEnd {$\Delta\xkl\ysp\leq\Delta\xkl^\tmax,\forall k,\forall l$}
      {$s=s^\tmax,$\label{step:Success}}
\State $s=s+1$
\EndWhile
\IfThenEnd {$p_k>p_k^\tmax \text{ or }  p_r>p_r^\tmax,\forall k,r$}
      {$iflag=1,$ \label{step:Infeasible}} 
\end{algorithmic}
\vspace{-.1cm}
\end{algorithm}\vspace{-.1cm}

\subsection{Feasibility of Problem}\label{subsec:Feasibility}
P\ref{op:0727_1915} can be infeasible for the given SINR targets $\gamma\xkl$. Therefore, the feasibility of P\ref{op:0727_1915} must be assured in the first stage before solving P\ref{op:0727_1915} in the second stage. There are basically two techniques in the literature: 1) the feasibility conditions are derived \cite{217}, and 2) the initial conditions are relaxed, e.g., the SINR targets are tested and the feasible targets are searched \cite{206,477}. Deriving the feasibility condition of P\ref{op:0727_1915} is challenging even for a simpler network and a simpler problem \cite{654}. In \cite{654}, power minimization under the worst stream SINR condition is studied for a multi-antenna relay network with a single transmitter and a receiver, and the direct links are neglected. Furthermore, in \cite{654}, the feasibility condition is derived based on the signal-to-interference ratio (SIR) instead of SINR. After assuming the targets are feasible based on the approximate feasibility condition derived from the SIR metric, the beamforming vectors are solved based on SINR. The approximate feasibility condition is more accurate in the high SNR regime, where noise can be neglected. On the other hand, testing the SINR targets and searching for the feasible SINR targets \cite{206,477} are as costly as solving problem P\ref{op:0727_1915}, i.e., both require high iteration numbers and high computational complexities per iteration, which severely impedes the cross-analysis due to the long simulation durations.

In this work, we adopt a new technique. We randomly initialize the transmit and relay beamforming vectors with full powers, i.e., ${p_k^{\tmax}}$, ${\forall k\in\mathcal{K}}$ and $p_r^{\tmax}$, ${\forall r\in\mathcal{R}}$ to determine feasible SINR targets with a high probability. The infeasible cases that make up a small  portion of   tests along with slow and fluctuant converging cases are filtered through step \ref{step:Infeasible} and the detection window as detailed earlier. Hence, by testing randomly and automatically generated feasible SINR targets for any permutation of the network parameters, the cross-analysis is executed systematically and extensively within short simulation durations. Since cross-analysis at this comprehensive level has not been performed in the literature yet, many new insights are revealed in this paper. For all simulations, the receive filters are also randomly initialized but normalized to unity since including the receiver power as another parameter
in the already large set of \mbox{cross-varying} network parameters substantially complicates the \mbox{cross-analysis} in Section \ref{sec:NumericalResults}.

In fact, the transmit and relay power constraints
\begin{align}\label{eqn:180829_1906}
p_k\!=\!&\sum_{l=1}^{d_k}\xu\xkl^H\xu\xkl\!=\!\!\sum_{l=1}^{d_k}\ttr(\xX\xkl)\leq p_k^\tmax \text{ and}\nonumber\\
p_r\!=\!&\sum_{k=1}^{K}\sum_{l=1}^{d_k}\!\xu\xkl^H\xR\xxrk\xu\xkl\!=\!\!\sum_{k=1}^{K}\sum_{l=1}^{d_k}\!\ttr(\xX\xkl\xR\xxrk)\leq p_r^{\prime\tmax},
\end{align}
respectively, where ${p_r^{\prime\tmax}\triangleq p_r^\tmax-\sigma_r^2\ttr(\xF_r\xF_r^H)}$, can be incorporated into P$(\xX\xkl)$. The Matlab script for CVX solution is given in Algorithm \ref{alg:1012_1111}.
However, instead of incorporating the power constraints into CVX as shown in Algorithm \ref{alg:1012_1111}, checking the power constraints  at step \ref{step:Infeasible} of Algorithm \ref{alg:0818_1524} has the following advantages. Firstly, the CVX algorithm runs faster without the incorporated power constraints, particularly for networks with high number of relays and relay antennas. Secondly, if the SINR targets are infeasible, the algorithm starts fluctuating, i.e., different SINRs that do not meet SINR targets are achieved over the iterations while the power constraints are still assured by the constraints incorporated in Algorithm \ref{alg:1012_1111}. Therefore, it is \mbox{time-consuming} to detect whether the fluctuation is due to infeasibility or it is a rare case where the algorithm still converges
after the fluctuation. These two cases significantly hinder the simulation durations and executing extensive cross-analysis in Section \ref{sec:NumericalResults} becomes impractical. As seen
in Section \ref{sec:NumericalResults}, in the worst case, our proposed distributed algorithm requires $35$ iterations on average. Therefore, instead of plugging the power constraints
\eqref{eqn:180829_1906} into the CVX optimization, letting the algorithm converge to the SINR targets and then checking the power constraints by step \ref{step:Infeasible} of Algorithm \ref{alg:0818_1524} can swiftly determine the feasibility of the targets.

\begin{algorithm}[t]\small
\caption{ Matlab script for locally solving P($\xX\xkl$) with incorporated power constraints at each stream in parallel via CVX.} \label{alg:1012_1111}
\begin{algorithmic}[1]
\Statex cvx\textunderscore begin sdp
\Statex \quad variable $\xX\xkl(M,M)$ hermitian
\Statex \quad minimize(trace$(\xX\xkl\xRsumk)$);
\Statex \quad subject to
\Statex \qquad trace$(\xX\xkl\Rklk)==\gamma\xkl(\zeta\xkl+\sigma_{n\xkl}^2)$;
\Statex \qquad trace$(\xX\xkl)<= p_k^\tmax-\sum_{m=1,m\neq l}^{d_k}\text{trace}(\xX\xkm)$;
\Statex \qquad for r=1:$R$
\Statex \qquad ~trace$(\xX\xkl\xR_{rk})<=p_r^{\prime\tmax}-\!\!\!\sum_{(j,m)\neq(k,l)}\text{trace}(\xX\xjm\xR_{rj})$;
\Statex \qquad end
\Statex \quad $\xX\xkl>=0$;
\Statex cvx\textunderscore end
\end{algorithmic}
\vspace{-.1cm}
\end{algorithm}

\subsection{Benchmark Distributed Algorithms}
\subsubsection{ADMM with Bounded Guarantee (ADMM-BG)}
The conventional dual decomposition method \cite{731} for a distributed solution of P\ref{op:0727_1915} is not applicable since the dual of P\ref{op:0727_1915} can be unbounded \cite{684}. This can be demonstrated by considering to solve the dual problem

\vspace{.2cm}P($\mu\xkl,\xX\xkl$) \vspace{-.3cm}
\begin{align}
&\hspace{1cm} \underset{\mu\xkl}\max~~\underset{\xX\xkl}\min & &\mkern-10mu \mathfrak{L}_{\rho\xkl}^{\prime \text{~Conv.}}\nonumber\\
&\hspace{1.05cm} \text{s.t.} & &\mkern-10mu \text{The constraint }\eqref{op:0727_1915_SPSDConstraint}^\prime,
\end{align}
where $$\mathfrak{L}_{\rho\xkl}^{\prime \text{~Conv.}}(\xX\xkl,\eta\xkl^\prime)=\ttr(\xX\xkl\xRsumk)-\eta\xkl^\prime\frac{1}{\gamma\xkl}\ttr(\xX\xkl\Rklk)$$
instead of solving the subproblem P($\xX\xkl$) in step a) of our proposed algorithm presented earlier via CVX. The inner optimization problem of the above dual problem can be unbounded below given the dual variable $\mu\xkl$ so that $-\eta\xkl^\prime\ttr(\xX\xkl\Rklk)\rightarrow-\infty$, i.e., the solution can go to minus infinity. This problem can be avoided by the extra quadratic penalty term added in the ADMM scheme
\begin{align}\label{eqn:180328_2248}
\mathfrak{L}_{\rho\xkl}^{\prime \text{ADMM}}(\xX\xkl,\eta\xkl^\prime,\zeta\xkl)=& \ttr(\xX\xkl\xRsumk)-\eta\xkl^\prime\frac{1}{\gamma\xkl}\ttr(\xX\xkl\Rklk)\nonumber\\
&\hspace{-2cm}+\frac{\rho}{2}\left(\zeta\xkl-\frac{1}{\gamma\xkl}\ttr(\xX\xkl\Rklk)+\sigma_{n\xkl}^2\right)^2.
\end{align}
As seen in \cite[Eq. (23)]{684}, authors enforce the quadratic penalty term also for the first term $\ttr(\xX\xkl\xRsumk)$ in \eqref{eqn:180328_2248} by defining an auxiliary variable
\begin{equation}\label{eqn:180301_1450}
p\xkl\triangleq\ttr(\xX\xkl\xRsumk),\forall k\in\mathcal{K},\forall l\in\mathcal{L}_k,
\end{equation}
and then introducing the slack variables $t\xkl\geq0,\forall k\in\mathcal{K},\forall l\in\mathcal{L}_k$ \cite[Eq. (29d)]{684}
\begin{equation}\label{eqn:0223_2004}
p\xkl=t\xkl.
\end{equation}
Note that the first term $\ttr(\xX\xkl\xRsumk)$ in \eqref{eqn:180328_2248} is the objective term in P\ref{op:0727_1915}, similar to \cite[Eq. (17a)]{684}. Hence, the final local partial augmented Lagrangian function is given as
\begin{align}
\mathfrak{L}_{\rho\xkl}^{\prime\prime \text{ADMM}}(\xX\xkl,\eta\xkl^{\prime\prime},\zeta\xkl,p\xkl,t\xkl)=&p\xkl\!-\!\eta\xkl^{\prime\prime}p\xkl\nonumber\\
&\hspace{-4cm}+\!\frac{\rho}{2}(t\xkl\!-\!p\xkl)^2\!-\!\eta\xkl^\prime\frac{1}{\gamma\xkl}\ttr(\xX\xkl\Rklk)\nonumber\\
&\hspace{-4cm}+\frac{\rho}{2}\left(\zeta\xkl-\frac{1}{\gamma\xkl}\ttr(\xX\xkl\Rklk)+\sigma_{n\xkl}^2\right)^2.
\end{align}
Following the same principle from \cite{684} as explained above, the modified local partial augmented Lagrangian function of \eqref{eqn:1803011340} is
\begin{align}
\mathfrak{L}_{\rho\xkl}(\eta\xkl^{\prime\prime},\zeta\xkl,\zeta\xkl^\backprime,p\xkl,t\xkl,\lambda\xkl)=&p\xkl-\eta\xkl^{\prime\prime}p\xkl\nonumber\\
&\hspace{-4cm}+\frac{\rho}{2}(t\xkl-p\xkl)^2\nonumber\\
&\hspace{-4cm}+\lambda\xkl(\zeta\xkl+\zeta\xkl^\backprime)+\frac{\rho}{2}(\zeta\xkl+\zeta\xkl^\backprime)^2,
\end{align}
where $\lambda\xkl$ is the notation used in \eqref{eqn:1803011340} that corresponds to $\eta\xkl^\prime$ used in this subsection. Hence step a) of our proposed distributed algorithm is modified as follows

\vspace{.2cm}P($p\xkl,\xX\xkl$) \vspace{-.3cm}
\begin{align}\label{eqn:180329_0113B}
&~~~ \underset{p\xkl,\xX\xkl}\min & &\mkern-10mu p\xkl-\eta\xkl^{\prime\prime}p\xkl+\frac{\rho}{2}(t\xkl-p\xkl)^2\nonumber\\
&\hspace{.7cm} \text{s.t.} & &\mkern-10mu \text{The constraints }\eqref{op:0727_1915_SPSDConstraint}^\prime,\eqref{op:0727_1915_SINR2ndConstraint}^\prime, \text{ and } \eqref{eqn:180301_1450}^\prime,
\end{align}
which can be again solved via CVX. Accordingly, the algorithm \ref{alg:0818_1524} needs small modifications. In summary, in \mbox{ADMM-BG}, P($p\xkl,\xX\xkl$) is solved instead of P($\xX\xkl$) in step a) of our proposed algorithm to guarantee bounded solutions.

Due to the added superfluous variables and constraints in \eqref{eqn:180301_1450} and
\eqref{eqn:0223_2004}, the modified approach, ADMM-BG algorithm, has a slower convergence rate, i.e., more iterations are needed, and also has a higher computational complexity per iteration than the original approach in step a), i.e., the proposed algorithm in this paper. As explained in more detail in Section \ref{sec:Attributes}, as the problem size increases, i.e., more constraints in this case, the complexity of problem solution increases since CVX depends on the number of constraints. The iteration number also increases as a direct consequence of a larger problem size.
Moreover, it is observed that \mbox{ADMM-BG} scheme is significantly more unstable, i.e., the rare frequencies of \mbox{slow-,} fluctuant-, and non-convergent cases significantly increase. Finally, more constraints directly necessitates more variables, i.e., more antennas at the transmitter side. Due to the constraint \eqref{eqn:180301_1450},  the minimum number of transmitter antennas is more for ADMM-BG.

\subsubsection{Accelerated Distributed Augmented Lagrangians (ADAL)}
In \cite{682}, an improved version of ADMM is proposed, which is coined as ADAL since it converges faster than the conventional ADMM. ADAL introduces an auxiliary variable for each interference term inside the summation \cite[Eq. (11), 3rd constraint]{682} in contrast to the summation of interference terms \eqref{op:0727_1915_SINR3rdConstraint} as proposed in this work. Hence, instead of two constraints \eqref{op:0727_1915_SINR2ndConstraint} and \eqref{op:0727_1915_SINR3rdConstraint} per stream as proposed in this work, ADAL has 1 (desired signal) + ${B-1}$ (interference signals)${=B}$ constraints per stream. More constraints increase 1) the iteration number (the problem size increases, similar to more number of users and antennas), 2) the computational complexity per iteration (the complexity of CVX is dependent on the number of constraints), and 3) the minimum number of antennas (the transmit covariance is involved in more constraints, thus more antennas-variables are needed). Following the approach in \cite{682}, the resulting SDR of P\ref{op:0727_1557} at the optimal point can be equivalently rewritten as

\vspace{.145cm}\problemcounter{}\label{op:180301_1800} \vspace{-.73cm}
\begin{subequations}
\begin{align}
&\mkern-15mu \underset{\{\xX\xkl,\xzeta\xkl\}_{\forall k,l}}\min & & \mkern-15mu \sum_{k=1}^{K}\sum_{l=1}^{d_k}\ttr(\xX\xkl\xRsumk)\\
& \hspace{.4cm} \text{s.t.} & & \mkern-40mu\sum_{k=1}^{K}\sum_{l=1}^{d_k}\xzeta\xkl=\xzero,\label{op:180301_1800_SINR1stConstraint}\\
& & & \mkern-40mu\zeta\xkl\ykl\!\!=\!\!\frac{1}{\gamma\xkl}\ttr(\xX\xkl\Rklk)\!-\!\sigma_{n\xkl}^2,\!\forall k\!\in\!\mathcal{K},\forall l\!\in\!\mathcal{L}_k\label{op:180301_1800_SINR2ndConstraint}\\
& & & \mkern-40mu\zeta\xjm\ykl\!\!=\!\!-\ttr(\xX\xkl\Rjmk),\forall k,j,\,j\!\neq\!k,\!\forall l,m,m\!\neq\! l\label{op:180301_1800_SINR3rdConstraint}\\
& & & \mkern-40mu \text{The constraint } \eqref{op:0727_1915_SPSDConstraint}, \nonumber
\end{align}
\end{subequations}
where ${\xzeta\xkl\triangleq[\zeta_{1,1}\ykl,\ldots,\zeta\xin,\ldots,\zeta_{K,d_K}\ykl]^T}$ is the ${B\times1}$ vector that stacks all auxiliaries $\zeta\xin\ykl$, ${\forall i\in\mathcal{K},\forall n\in\mathcal{L}_i}$.  The ADAL-Direct algorithm proposed in \cite{682} is shortly referred as ADAL in our work.

The partial augmented Lagrangian for P\ref{op:0727_1915} can be written as
\begin{align}
&\mathfrak{L}_\rho\big(\{\xX\xkl,\xzeta\xkl\}_{\forall k,l},\xlambda\big)\nonumber\\
&=\sum_{k=1}^{K}\sum_{l=1}^{d_k}\!\ttr(\xX\xkl\xRsumk)
+\xlambda^T\sum_{k=1}^{K}\sum_{l=1}^{d_k}\xzeta\xkl+\frac{\rho}{2}\left\Vert\sum_{k=1}^{K}\sum_{l=1}^{d_k}\xzeta\xkl\right\Vert_2^2,
\end{align}
where $\xlambda\triangleq[\lambda_{1,1},\ldots,\lambda_{K,d_K}]^T$ is the vector of Lagrange multipliers of the constraint \eqref{op:0727_1915_SINR1stConstraint}, $\rho\in\mathbb{R}_{+}$ is the Lagrangian dual update step size, and $\|.\|_2$ is the $L_2$ vector norm operator. ADAL proposes to distribute the problem over streams as follows, i.e., the local partial augmented Lagrangian is given as
\begin{align}\label{eqn:180328_2341}
&\mathfrak{L}_{\rho\xkl}\big(\xX\xkl,\xzeta\xkl,\vec{\xzeta}\xjm,\xlambda\big)\nonumber\\
&=\ttr(\xX\xkl\xRsumk)+\xlambda^T\xzeta\xkl+\frac{\rho}{2}\left\Vert\xzeta\xkl+\!\!\!\!\!\!\sum_{(j,m)\neq(k,l)}\vec{\xzeta}\xjm\right\Vert_2^2,
\end{align}
where ${\vec{\xzeta}\xjm}$ is the auxiliary variable transmitted from stream ${(j,m),\forall j\neq k,\,\forall m\neq l}$ to all streams ${(i,n)}$, ${\forall i\in\mathcal{K}\setminus \{j\}}$, ${\forall n\in\mathcal{L}_i\setminus\{m\}}$. In other words, the ${\vec{\xzeta}\xjm}$ variable is evaluated at stream ${(j,m)}$ and is transmitted to all streams ${(i,n),\forall i\neq j,\,\forall n\neq m}$, thus it is a constant vector for stream ${(k,l)}$.

\subsubsection*{The Main Steps of ADAL}
\paragraph{Updates of $\xzeta\xkl$ and $\xX\xkl$}\mbox{} The smaller problem of $\xzeta\xkl$ and $\xX\xkl$ is

\vspace{.2cm}P($\xzeta\xkl,\xX\xkl$) \vspace{-.2cm}
\begin{align}\label{eqn:180329_0113C}
&\underset{\xzeta\xkl,\xX\xkl}\min & &\mathfrak{L}_{\rho\xkl}\nonumber\\
& \hspace{.3cm} \text{s.t.} & & \text{The constraints }\eqref{op:0727_1915_SPSDConstraint}^\prime, \eqref{op:180301_1800_SINR2ndConstraint}^\prime,  \text{ and } \eqref{op:180301_1800_SINR3rdConstraint}^\prime,
\end{align}
where \eqref{op:180301_1800_SINR3rdConstraint}$^\prime$ this time denotes the constraint \eqref{op:180301_1800_SINR3rdConstraint} without $\forall j,m$; that is, for a particular stream $(k,l)$. For instance, assume $K=3$ and $d=2$, then for stream $(k,l)=(1,2)$, the set of $\zeta$ variables is $Z=\{\zeta_{1,1}^{1,2},\zeta_{1,2}^{1,2},\ldots,\zeta_{3,1}^{1,2},\zeta_{3,2}^{1,2}\}$, where ${|Z|=B}$. The above linear problem can be solved locally at each processor in parallel via CVX in Matlab.
\paragraph{Update of $\vec{\xzeta}\xkl$}
The $\xzeta\xkl^s$ output of P($\xzeta\xkl,\xX\xkl$) is further updated as follows
\begin{equation}\label{eqn:180327_2345a}
\vec{\xzeta}\xkl\ysp=\vec{\xzeta}\xkl^s+\tau(\xzeta\xkl^s-\vec{\xzeta}\xkl^s),
\end{equation}
where $\tau$ is a step size. $\vec{\xzeta}\xkl\ysp$ is to be transmitted to the other nodes as mentioned earlier.
\paragraph{Update of $\xlambda$}
As the final step, the dual variables are updated as follows
\begin{equation}\label{eqn:180327_2345b}
\xlambda\ysp=\xlambda^s+\tau\rho\sum_{k=1}^{K}\sum_{l=1}^{d_k}\vec{\xzeta}\xkl\ysp.
\end{equation}
Note that the above update is also local to a processing node $(k,l)$ since $\vec{\xzeta}\xjm\ysp, \forall(j,m)\neq(k,l)$ updates are provided from other nodes than $(k,l)$.

The ${B\!-\!1}$ constraints in \eqref{op:180301_1800_SINR3rdConstraint}$^\prime$ of P($\xzeta\xkl,\xX\xkl$) are overwhelming as opposed to our proposed ADMM algorithm, where the summation of these terms is utilized. Hence, the disadvantages of ADAL follow as mentioned earlier.

\begin{remark}
$\xzeta\xkl$ and $\xX\xkl$ variables in P($\xzeta\xkl,\xX\xkl$) are jointly solved via CVX.  the closed-form solution of $\xzeta\xkl$ can be obtained from the KKT conditions, then only the $\xX\xkl$ variables are solved via CVX to reduce the overall computational complexity.
\end{remark}
 Consider whether the joint solution of $\xzeta\xkl$ and $\xX\xkl$ variables via CVX or closed-form solution of $\xzeta\xkl$ and CVX solution of $\xX\xkl$ option should be chosen. As mentioned earlier, there are $B$ equality constraints per stream due to the \eqref{op:180301_1800_SINR2ndConstraint}$^\prime$ and \eqref{op:180301_1800_SINR3rdConstraint}$^\prime$ constraints. With $\zeta\xkl\yin$ (joint solution) each constraint is a bivariate equation of $\zeta\xkl\yin$ (scalar) and $\xX\xkl$ (matrix), whereas without $\zeta\xkl\yin$ (no joint solution), each equality is a univariate equation of $\xX\xkl$. Thus, in the univariate case, $B$ is the minimum number of antennas required at a node, and in the bivariate case, the requirement is lower than $B$ by the help of additional variables $\zeta\xkl\yin$. As shown earlier, our proposed method has significantly lower number of equality constraints, thus both the iteration number (can be determined from numerical results since the computation of iteration number is not possible) and computational complexity (can be approximately computed in terms of limiting behavior) per iteration are significantly low. Therefore, we choose joint solution of $\xzeta\xkl$ and $\xX\xkl$ via CVX for ADAL so that the iteration number comparison with our proposed distributed algorithm and the number of antennas in Section \ref{sec:NumericalResults} is more fair and convenient to compare.

From the above discussion, it can be also concluded that centralized algorithms can be less restrictive than the distributed algorithms in terms of minimum number of antenna requirements in general. The SDR of the QCQP problem P\ref{op:0727_1557} can be solved via CVX \cite{198} in a centralized manner. Hence, for the centralized algorithm, there are $Kd$ constraints and $KM$ variables. On the other hand, for ADAL, there are ${Kd+1\approx Kd}$ constraints and $M$ variables per processor. Therefore, the number of variables for the distributed algorithm is approximately shrunk $1/K$ times compared to the centralized solution. However, a centralized solution is less attractive when the dimension of the problem is large as mentioned earlier.

Similarly, it can be concluded that the power control algorithms are also more restrictive than the beamforming optimization problems. In particular, there is only 1 power variable $p\xkl$ per stream, on the other hand, there are $M_k$ variables per stream in the transmit beamforming vector $\xu\xkl\in\mathbb{C}^{M_k}$. Since the number of variables in the beamforming vectors that directly determine the SINR targets is significantly larger than the number of power variables, automatic generation of feasible SINR targets by random initializations of beamforming vectors is not possible for power control algorithms.

\subsection{Competitive SINR Targets}

P\ref{op:0720_1620}, where  the objective function is minimizing the total power and the QoS guarantee is meeting the SINR targets, is one of the most important optimization problems for resource allocation in wireless networks \cite{684,713,682,721,714}. This problem can be considered as the dual interest of network utility maximization problems \cite{740}, \cite[Chap. 8]{214}.
Depending on the system design goals, the trade-off between the consumed total power and the achieved sum-SINR in the network can be controlled by the given SINR targets.

The SINR targets can be determined at the physical or higher layers of the network. One of the approaches in the physical layer is to search for feasible SINR targets between the
lower and upper bounds set by beamforming filters. For the lower bound, random filter initializations are good choices as supported by the numerical results in this paper. For
the upper bound, SINR maximizing (\mbox{max-SINR}) filter initializations can be good choices as shown in \cite{55}. In Section \ref{subsec:SINRtargets}, we present the numerical results showing that by searching for higher feasible SINR targets between the lower and upper bounds, the system can achieve higher SINRs at the cost of increased power consumptions.
In Appendix \ref{app:DIA}, closed-form solutions of max-SINR filters for distributed implementation are provided.

\section{Attributes of the Algorithms}\label{sec:Attributes}
\subsection{Convergence}
The global convergence of two-block ADMM is proven based on mild conditions in \cite{726}. The optimal $\xX\xkl$ solution for each subproblem of P\ref{op:0727_1915} is a global optimal solution for the complete SDR problem in P\ref{op:0727_1915} since at each iteration, each subproblem of P\ref{op:0727_1915} is convex, and the iterations are pursued until convergence. We analytically and numerically show in this section and in Section \ref{sec:NumericalResults}, respectively, that the proposed distributed algorithm is convergent.

The convergence of ADMM to the global optimum of P($\xx,\xy$) under mild conditions is given by the following lemma.
\begin{lemma} \cite[Proposition 4.2]{719}\label{lem:180327_2237}
Assuming that the optimal solution set of P($\xx,\xy$) is nonempty, and $\xA^T\xA$ and $\xB^T\xB$ are invertible, the sequence of solutions $\{\xx\ysp,\xy\ysp,\xz\ysp\}$ obtained from \eqref{op:ADMMSola}, \eqref{op:ADMMSolb}, and \eqref{op:ADMMSolc}, respectively, is bounded and every limit point of $\{\xx\ysp,\xy\ysp\}$ is an optimal solution of P($\xx,\xy$).
\end{lemma}

\begin{proposition}
The iterates $\{\xX\xkl\ysp,\zeta\xkl\ysp,\zeta\xkl^{\backprime\, s+1}\}_{\forall k,l}$ and $\{\lambda\xkl\ysp,\mu\xkl\ysp,\mu\xkl^{\backprime s+1}\}_{\forall k,l}$ in Algorithm \ref{alg:0818_1524} converge to the optimal solution of primal and dual problems as $s\rightarrow\infty$. When the algorithm converges, the optimal $\{\xX\xkl\}_{\forall k,l}$ solutions obtained in step \ref{step:ADMM1st} of Algorithm \ref{alg:0818_1524} are the global optimal solutions of P\ref{op:0727_1915}.
\end{proposition}
\begin{proof}
Since ${g(\xy)=0}$ and ${\xA=\xzero}$, the Lagrangian function \eqref{op:ADMMSampleLagrangian} is modified correspondingly, and also, the ${\xA^T\xA}$ condition in Lemma \ref{lem:180327_2237} is dropped. Based on Lemma \ref{lem:180327_2237} and the correspondences in \eqref{eqn:180329_0047}, the proposition on the convergence of the proposed distributed algorithm follows.
\end{proof}
\subsection{Computational Complexity}\label{subsec:CompComplexity}
The most computationally intensive part of the proposed ADMM algorithm is step \ref{step:ADMM1st} of Algorithm \ref{alg:0818_1524}, where the general-rank positive semi-definite matrices $\{\xX\xkl\}$ are obtained explicitly via CVX without closed-form solutions. Obtaining closed-form beamforming solutions is challenging even in simpler network architectures \cite{684,682}. After obtaining the optimal general-rank positive semi-definite matrices $\{\xX\xkl\}$ via the proposed ADMM algorithm, rank-one solutions $\xu\xkl$, where \mbox{$\xu\xkl\xu\xkl^H=\xX\xkl$}, can be efficiently derived \cite{728}. Although there are no closed-form solutions, semi-definite programming (SDP) has still many real-time applications \cite{729,730}.
\begin{remark}
Similar to the subproblem of proposed distributed algorithm P($\xX\xkl$) in \eqref{eqn:180329_0113A}, note that the SDP solutions, e.g., via CVX, are also needed for the subproblems of benchmark schemes ADMM-BG P($p\xkl,\xX\xkl$) in \eqref{eqn:180329_0113B} and ADAL P($\xzeta\xkl,\xX\xkl$) in \eqref{eqn:180329_0113C}.
\end{remark}
The complexity of SDP solution for a generic SDR of a QCQP problem  \cite[Eq. (5)]{728} is given as $\mathcal{O}(\max\{M,c\}^4M^{1/2}\log(1/\epsilon))$,
where $M$ is the antenna number, $c$ is the number of constraints, and ${\epsilon>0}$ is the solution accuracy. Note that $c=1,2,\text{and }B$ for the proposed algorithm, ADMM-BG, and ADAL, respectively. Hence, for the former two algorithms, ${M\geq c}$, while for the last algorithm, ${M\geq c}$ or ${M<c}$. For the sake of simplicity, we evaluate the complexity of algorithms based on $\mathcal{O}(Mc)$; that is, the limiting behavior of simple multiplication of the number of variables and constraints. The complexity benchmark of algorithms based on this simplified formula also matches with CPU time based simulation benchmarks as shown in the end of next section.

For the proposed distributed algorithm, $P(\xX\xkl)$ has $M$ variables and $1$ constraint, P($\zeta\xkl,\zeta\xkl^\backprime$) has $2$ variables that have closed-form solutions \eqref{eqn:0818_1542}, and finally $3$ more variables in \eqref{eqn:0818_1548}, thus the complexity of proposed distributed algorithm is ${\mathcal{O}(M+5)}$ per processor. For ADMM-BG, P($p\xkl,\xX\xkl$) has $M$ (from $\xX\xkl$) $+$ $1$ (from $p\xkl$)  variables and $2$ constraints. Due to the similar steps to our proposed algorithm, ADMM-BG also has ${2+3=5}$ variables, and additional $1$ variable $t\xkl$. Hence, the complexity of ADMM-BG is ${\mathcal{O}(2(M+1)+6)}$. Finally, for ADAL, P($\xzeta\xkl,\xX\xkl$) has $M+B$ variables and $B$ constraints, and $B$ variables in each of \eqref{eqn:180327_2345a} and \eqref{eqn:180327_2345b}. Thus, the complexity of ADAL is ${\mathcal{O}(B(M+B)+2B)}$. In summary, the proposed algorithm has the lowest complexity, followed by ADMM-BG, and then ADAL. The numerical CPU benchmarks of the proposed and existing algorithms are presented in Section \ref{subsec:vsExisting}.
\subsection{Message Exchange Load}
For the update of \eqref{eqn:0818_1548_mub} at each processor, a single value is needed, i.e., the summation term, which can be delivered by a collector node. This means that each processor needs to send out ${B\!-\!1}$ scalar values, one for each other processor, to the collector node. In total, ${B(B\!-\!1)}$ and $B$ numbers of scalar values need to be exchanged from the processors to the collector node and from the collector node to the processors, respectively. As a result, in total, $B^2$ scalar values need to be exchanged in the network. Since the update  \eqref{eqn:0818_1548_mub} is also needed for ADMM-BG, the message exchange loads of ADMM-BG and the proposed algorithm are same. For ADAL, the updates \eqref{eqn:180327_2345b} are achieved by sending out \eqref{eqn:180327_2345a} from each stream to the collector node. Therefore, for both message exchange directions between the collector node and processors, exchange of $B^2$ scalar values is needed, making a sum of $2B^2$ exchange of scalar values in the network. In summary, ADAL has the highest message exchange load, followed by ADMM-BG and the proposed algorithm.

Note that a deployed central node or any node among the existing nodes of relay network can serve as a collector node. Moreover, parallel updates of the variables in the ADMM algorithm are robust to delays and errors \cite{695,727}.

\section{Distributed Joint Transmit and Relay Beamforming Filter Optimization} \label{sec:Joint}
Relay filter design has been a long standing open problem when the direct links exist. To pinpoint the hurdle that direct links cause in relay filter design,  SINR \eqref{eqn:180328_1610} is reformulated as an explicit function of relay filters next.

\subsection{SINR Reformulation}
Consider the signal power from stream $n$ of transmitter $i$ to stream $l$ of receiver $k$
\begin{align}\label{eqn:171123_1150}
p_{klin}\!\!=&\xu\xin^H\!\!\left[\xJ\xxki^H~\xH^{\prime H}\xkRi\right]{\!\!\!\left[\begin{array}{c}\xv\xkl(1) \\  \xv\xkl(2) \\ \end{array}\right]}
\!\!\!\left[\xv\xkl^H(1)\xv\xkl^H(2)\right]{\!\!\!\left[\begin{array}{c}\xJ\xxki \\ \xH^{\prime}\xkRi \\ \end{array}\right]}\!\xu\xin \nonumber\\
=&\xu\xin^H\xJ\xxki^H\xv\xkl(1)\xv\xkl^H(1)\xJ\xxki\xu\xin\nonumber\\
&+\xu\xin^H\xH^{\prime H}\xkRi\xv\xkl(2)\xv\xkl^H(2)\xH^{\prime}\xkRi\xu\xin\nonumber\\
&+\xu\xin^H(\xA^H+\xA)\xu\xin,
\end{align}
where $${\xA\triangleq\xJ\xxki^H\xv\xkl(1)\xv\xkl^H(2)\xH^{\prime}\xkRi},$$
${\xv\xkl(1)}$ and ${\xv\xkl(2)}$ is the \mbox{$M_k\times1$} receive beamforming vector for the $l\yth$ stream of the $k\yth$ user in the first and second time slot, respectively.

The first and second summand in the last equality of \eqref{eqn:171123_1150} is independent of relay filters and can be rewritten in terms of Hermitian matrices.  ${\xA^H+\xA}$ in the third summand is a Hermitian matrix, thus the third summand can be rewritten as ${\tre(2\xu\xin^H\xA\xu\xin)}$, where
\begin{equation}
  \xu\xin^H\xA\xu\xin\!\!=\xu\xin^H\xJ\xxki^H\xv\xkl(1)\xv\xkl^H(2)\sum_{r=1}^R\xG\xxkr\xF_r\xH_{ri}^{\prime\prime}\xu\xin.
\end{equation}
Hence, the signal power from stream $n$ of transmitter $i$ to stream $l$ of receiver $k$ via relay $r$ is given as
\begin{equation}
p\xklrin=2\tre\left(\ttr(\xQ\xklrin\xF_r)\right),
\end{equation}
where $$\xQ\xklrin\triangleq\xH_{ri}^{\prime\prime}\xu\xin\xu\xin^H\xJ\xxki^H\xv\xkl(1)\xv\xkl^H(2)\xG\xxkr$$ is a \mbox{non-Hermitian matrix}. Due to the asymmetry that the direct links bring through the third summands of signal powers in \eqref{eqn:171123_1150}, relay filter design in the existence of direct link is a \mbox{non-trivial} problem. The third
summand is particularly not negligible when the direct channels $\xJ\xxki$ and the effective channels $\xH^{\prime}\xkRi$ are correlated. However, when these channels are independent, the affect of this term is less prominent. By omitting this term, the approximate SINR of the $l\yth$ stream of the $k\yth$ user is obtained as
\begin{equation}\label{eqn:171124_1102}
\widehat{\tSINR}\xkl(\xu\xin)\!\!=\!\frac{\zeta\xkl\ykl(1)\!+\!\zeta\xkl\ykl(2)}{\sum_{(j,m)\neq(k,l)}\!(\zeta\xkl\yjm(1)\!+\!\zeta\xkl\yjm(2))\!+\!\sigma_{n\xkl}^2},
\end{equation}
where
\begin{equation*}
\zeta\xkl\yin(1)\!\!\triangleq\!\big|\xv\xkl^H(1)\xJ\xxki\xu\xin\big|^2, \zeta\xkl\yin(2)\!\!\triangleq\!\big|\xv\xkl^H(2)\xH^{\prime}\xkRi\xu\xin\big|^2, \text{and}
\end{equation*}
  $\sigma_{n\xkl}^2$ is the sum of noises in the first and second time slots that can be obtained from \eqref{eqn:180515_1330}.

Next, \eqref{eqn:171124_1102} can be rewritten as an explicit function of relay filter
\begin{equation}\label{eqn:171212_1640}
\widehat{\tSINR}\xkl(\xf_r)=\frac{\bar{\xf}_r^H\tilde{\xC}_{klrkl}\bar{\xf}_r}{\bar{\xf}_r^H\tilde{\xC}_{klrjm}\bar{\xf}_r}=\frac{\ttr(\xY_r\tilde{\xC}_{klrkl})}{\ttr(\xY_r\tilde{\xC}_{klrjm})},
\end{equation}
where ${\xf_r\triangleq\tvec(\xF_r)},$ $\bar{\xf}_r\triangleq\left[
  \begin{array}{c}
    \xf_r \\
    t \\
  \end{array}
\right]$, ${|t|^2=1}$, and ${\xY_r\triangleq\bar{\xf}_r\bar{\xf}_r^H}$. Further derivation details of \eqref{eqn:171212_1640} and the definition of $\tilde{\xC}_{klrin}$ are given in Appendix \ref{app:SINRApprox}.

\subsection{Total Power Reformulation}
Thus far, the exact SINR \eqref{eqn:SINR} is approximated by \eqref{eqn:171124_1102}, and the approximate SINR is rewritten as a function of relay filter in homogenous quadratic form \eqref{eqn:171212_1640}.
Next, the total power is rewritten as a function of relay filters.

Noting the relay power \eqref{eqn:180321_PC2} and using the equality $\ttr(\xA\xB\xA^H)=\tvec(\xA)^H(\xB\otimes\xI)\tvec(\xA)$, the total power is obtained as\begin{equation}\label{eqn:180829_1430a}
\sum_{r=1}^R\xf_r^H\xD_{r.}\xf_r,
\end{equation}
where
\begin{align*}
&\xD_{r.}\!\!\!\!&\triangleq&\bigg(\sum_{k=1}^K\sum_{l=1}^{d_k}\xD_{rkl}+\sigma_r^2\xI_{N_r}\bigg)^T\!\!\!\otimes\xI_N=\sum_{k=1}^K\sum_{l=1}^{d_k}\xD_{rkl}^\prime,\\
&\xD_{rkl}\!\!\!\!&\triangleq&\xH^{\prime\prime}_{rk}\xu\xkl\xu\xkl^H\xH_{rk}^{\prime\prime H} \text{, and}\\
&\xD_{rkl}^\prime\!\!\!\!&\triangleq&\bigg(\xD_{rkl}+\frac{\sigma_r^2\xI_{N_r}}{B}\bigg)^T\otimes\xI_{N_r}.
\end{align*}
Furthermore, \eqref{eqn:180829_1430a} can be rewritten as
\begin{equation}\label{eqn:180829_1430b}
\sum_{r=1}^R\sum_{k=1}^K\sum_{l=1}^{d_k}\ttr(\xY_r\bar{\xD}_{rkl}^\prime),
\end{equation}
where $$\bar{\xD}_{rkl}^\prime\triangleq\Bigg[
\begin{array}{cc}
\xD_{rkl}^\prime+(\rho_2/2)\xI_{N_r^2} & \xzero_{\!N^2\!-\!1\!\times\!1} \\
\xzero_{1\!\times\!N^2\!-\!1}  & 0
\end{array}
\Bigg]_{N^2+1\times N^2+1}.$$

\subsection{Problem Formulation}
Using \eqref{eqn:171212_1640} and \eqref{eqn:180829_1430b}, the relay beamforming filter design problem is given as

\vspace{6cm}\problemcounter{}\label{op:180829_1252} \vspace{-.5cm}
\begin{subequations}
\begin{align}
&~ \underset{\{\xY_r\}_{\forall r}}\min & & \sum_{r=1}^R\sum_{k=1}^{K}\sum_{l=1}^{d_k}\ttr(\xY_r\bar{\xD}_{rkl}^\prime) \label{op:171226_1527a}\\
&\hspace{.3cm} \text{s.t.} & & \dot{\zeta}_r\ykl+\dot{\zeta}_r^{k,l\, \backprime}=0,\forall k\!\in\!\mathcal{K},\,\forall l\!\in\!\mathcal{L}_k,\forall r\!\in\!\mathcal{R}\label{op:171226_1527b}\\
&\mkern-20mu & &\mkern-30mu\dot{\zeta}_r\ykl\!=\!\frac{1}{\gamma\xkl}\ttr(\xY_r\tilde{\xC}_{klrkl}),\forall k\!\in\!\mathcal{K},\,\forall l\!\in\!\mathcal{L}_k,\forall r\!\in\!\mathcal{R}\label{op:171226_1527c}\\
&\mkern-20mu & &\mkern-30mu\dot{\zeta}_r^{k,l\,\backprime}\!=\!-\ttr(\xY_r\tilde{\xC}_{klrjm}),\forall k\!\in\!\mathcal{K},\,\forall l\!\in\!\mathcal{L}_k,\forall r\!\in\!\mathcal{R}\label{op:171226_1527d}\\
&\mkern-20mu & &\mkern-30mu\ttr(\xY_r\sum_{k=1}^{K}\sum_{l=1}^{d_k}\bar{\xD}_{rkl}^\prime)\leq p_r^\tmax,\forall r\!\in\!\mathcal{R}\label{op:171226_1527e}\\
&\mkern-20mu & &\mkern-30mu\ttr(\xY_r\xTheta)=1,\forall r\!\!\in\!\!\mathcal{R}\\
&\mkern-20mu & &\mkern-30mu\xY_r\!\in\!\mathbb{S}_+^N,\forall r\!\in\!\mathcal{R}\label{op:171226_1527f},
\end{align}
\end{subequations}
where $\xTheta\triangleq\left[
\begin{array}{cc}
\boldsymbol{0} & \boldsymbol{0} \\
\boldsymbol{0} &1\\
\end{array}
\right].$
 The partial augmented Lagrangian is given as
\begin{align}
&\mathfrak{L}_{\dot{\rho}}\big(\{\xY_r,\dot{\zeta}_r\ykl,\dot{\zeta}_r^{k,l\, \backprime},\dot{\lambda}_r\ykl\}_{\forall k,l,r}\big)\nonumber\\
&=\sum_{k=1}^K\sum_{l=1}^{d_k}\sum_{r=1}^R\Big(\ttr(\xY_r\bar{\xD}_{rkl}^\prime)+\dot{\lambda}_r\ykl(\dot{\zeta}_r\ykl+\dot{\zeta}_r^{k,l\, \backprime})\nonumber\\
&\hspace{2.5cm}+\frac{\dot{\rho}}{2}(\dot{\zeta}_r\ykl+\dot{\zeta}_r^{k,l\, \backprime})^2\Big),
\end{align}
where $\dot{\lambda}_r\ykl\in\mathbb{R}$ is the Lagrange multiplier of the constraint \eqref{op:171226_1527b}, and $\dot{\rho}\in\mathbb{R}_{+}$ is the Lagrangian dual update step size. From hereafter, the distributed solution follows the similar steps of the distributed transmit beamforming filter design proposed in Section \ref{subsec:ProposedADMM}.
\subsection{Pseudocode}

The algorithm for the distributed joint solution is obtained by adding two more steps after step 5 of Algorithm \ref{alg:0818_1524}. Similar to the transmitter side optimization,
the first step after step 5 of Algorithm \ref{alg:0818_1524} is obtaining $\{\xY_r\}$. The second step is obtaining the auxiliary variables $\{\dot{\zeta}_r\ykl,\dot{\zeta}_r^{k,l\,\backprime}\}$, and the Lagrangian multipliers $\{\dot{\mu}\xkl\}$ and $\{\dot{\mu}\xkl^{k,l\,\backprime}\}$ for the constraints \eqref{op:171226_1527c} and \eqref{op:171226_1527d}, respectively.

The discussions of distributed transmit beamforming filter design hold for the distributed relay beamforming filter design as well. For instance, CVX optimization and closed-form solutions are utilized
for $\{\xY_r\}$ and $\{\dot{\zeta}_r\ykl,\dot{\zeta}_r^{k,l\,\backprime},\dot{\mu}\xkl,\dot{\mu}\xkl^{k,l\,\backprime}\}$, respectively, and the optimal $\xY_r$ matrices are observed
to be always \mbox{rank-one}. However, in contrast to the transmit side optimization as discussed at the end of Section \ref{subsec:Feasibility}, the relay power constraint \eqref{op:171226_1527e} is better to be incorporated into the CVX optimization, and the transmit power constraint \eqref{eqn:180829_1906}
at the transmitter side as well, since the distributed joint algorithm is significantly more costly in computations due to
the distributed relay filter optimization
 as shown in Section \ref{subsec:DistributedJoint}.

\section{Numerical Results} \label{sec:NumericalResults}

In this section, the performance of the proposed scheme is compared with the centralized, ADMM-BG, and ADAL solutions by cross-varying the network parameters. All distributed algorithms in this section are optimal, thus they achieve the same SINR targets and also have the same total power consumptions that are achieved by the centralized solution. The network parameters to be cross-varied are the number of transmit ($M$) and relay ($N$) antennas, the number of relays ($R$), and the transmit ($\tSNR_\tT$) and relay ($\tSNR_\tR$) SNRs.

\subsection{Simulation Settings}

In our simulations, the distances between transmitter and relay sides, and relay and receiver sides are same and fixed to 1 km. The shadowing coefficient follows log-normal distribution with zero mean and standard deviation equal to $8$. The power of receive filter $\bar{\xv}\xkl$ is normalized to 1 by equal power allocation over the receive filters $\xv\xkl(i),\,i=1,2$, i.e., the receive filter power for each time slot is normalized to $\sqrt{0.5}$. The transmitters and relays have the same transmit power constraints among themselves, i.e., $p_k^{\tmax}=p_\text{T}^{\tmax}$, $\forall k \in \mathcal{K}$ and $p_r^{\tmax}=p_\text{R}^{\tmax}$, $\forall r \in \mathcal{R}$, where ${\mathcal{R}\triangleq\{1,2,\ldots,R\}}$ is the set of relays, but $p_\text{T}^{\tmax}$ and $p_\text{R}^{\tmax}$ are not necessarily equal. Hence, the SNR at the transmitter and relay side can be defined as
\begin{equation}\label{eqn:MaximumPowers}
\tSNR_\text{T}=\frac{p_\text{T}^\tmax}{\sigma^2} \text{ and } \tSNR_\text{R}=\frac{p_\text{R}^\tmax}{\sigma^2},
\end{equation}
respectively, and $\sigma_k^2=\sigma^2,\,\forall k\in\mathcal{K}$ is assumed for simplicity. Without loss of generality, the streams of a user are assumed to have the same SINR target, which is  equal to the average SINR of the user, i.e., the average of user's stream SINRs. For both the proposed algorithm and \mbox{ADMM-BG}, the step sizes $\rho$ and $\rho_c$ are empirically tuned to $1.2$ and $0.5$, respectively. For ADAL, the step sizes $\rho$, $\rho_c$, and $\tau$ are tuned to $9$, $0.5$, and $0.3$, respectively. The step sizes are kept constant for the simulations. If $|\tSINR\xkl-\gamma\xkl|\leq\Delta\xkl^\tmax=10^{-4}$ is met by the stream, the stream is regarded as achieving the SINR target.

Finally, for the benchmarks of the proposed algorithm and centralized solution, $100$ Monte Carlo channels are tested. Due to the slow convergence rates of ADMM-BG and ADAL, $20$ Monte Carlo channels are tested for their benchmarks with the proposed algorithm.

\subsection{Some Notes}

For all simulations, random initializations of filters and channels are pre-saved in files to feed the same inputs to the algorithms. Thus, reliable conclusions with less number of channel tests can be reached. Especially for large network sizes, this approach significantly reduces the simulation durations. Therefore, the networks with many varying parameters are tested in reasonable time durations by the use of pre-saved inputs.

The numerical results in this section are presented for ${K\!=\!3}$ networks since cross-varying $K$ as well significantly increases the data sets to be analyzed. In all tests, each user is assumed to have \mbox{$d_k=2$} streams. Thus, a total of ${B=6}$ streams are transmitted in all network configurations. Inline with the objective function of P\ref{op:0720_1620}, the total power curves are plotted in this section. On the other hand, although  QoS assurance of P\ref{op:0720_1620} imposes an SINR target for each of the
 streams, a  \mbox{single} curve, sum-SINR, is plotted instead of plotting $6$ individual stream SINR curves for the sake of clarity of the figures.

 For all simulations of the proposed distributed algorithm presented in this section excluding the Section \ref{subsec:SINRtargets}, the percentage of all
cases including slow-, fluctuant-, and non-convergent, i.e., infeasible SINR targets, makes up in total $1.93\%$. We note that in Section VI-F, we propose the linear search method to determine higher feasible SINR targets than those are set by random filter initializations. The next iteration of linear search method is proceeded depending on whether the result of the previous iteration is feasible or infeasible.

The conclusions drawn in this section by cross-varying these parameters are not straightforward, e.g., increasing $R$ can assist in achieving the SINR targets while also increasing the power consumption. Moreover, as discussed later, the ratios of these parameters also significantly affect the outputs, e.g., the ratios of $\tSNR_\tT$ and $\tSNR_\tR$ have different effects on the outputs.

\subsection{Proposed ADMM vs. Centralized Solution (Figs. \ref{Fig:TPSSINvsNS} -- \ref{Fig:INvsNS_12424212})}
\subsubsection{Performances over Network Sizes and SNRs (Fig. \ref{Fig:TPSSINvsNS})}
In Figs. \ref{Fig:TPvsNS}, \ref{Fig:SSvsNS}, and \ref{Fig:INvsNS_12122121}, the numerical results of average total power, sum-SINR, and iteration number vs. $M$ and $R$ are plotted, respectively.

As seen in Fig. \ref{Fig:TPvsNS}, the proposed algorithm achieves the optimal centralized algorithm solutions, i.e., they have the same power consumptions. As mentioned earlier, all algorithms achieve the given SINR targets as seen in Fig. \ref{Fig:SSvsNS}, and the optimality of an algorithm is determined whether it has the same power consumption with the centralized solution or not. In all cases, the numerical sum-SINR results in Section \ref{sec:NumericalResults} equal to the sum of the given feasible SINR targets of each stream, which are determined by random initializations of beamforming vectors.

As seen in Figs. \ref{Fig:TPvsNS} and \ref{Fig:SSvsNS}, the network $\mIVnVIIIrIX$ has a higher sum-SINR and a lower total power consumption than $\mIIInVIIIrX$, respectively. Since the optimization problem, power minimization while achieving the SINR targets, is solved at the transmitter side, providing more resources, i.e., more antennas, to the transmitters returns better results, i.e., higher SINRs are achieved with lower power consumptions. As expected, the sum-SINRs of the  networks $\mXnVIIIrIII$, $\mXnVIIIrIV$, and $\mXVnVIIIrX$ gradually increase as the total numbers of relays in the networks increase. However, their total power consumptions remain similar because although $R$ is increased, $M$ is chosen sufficiently large for these networks.

As seen in Fig. \ref{Fig:INvsNS_12122121}, the iteration number increases when the SNR increases. When the transmitter and relay SNRs are $12$ (${\tSNR_\tT\!=\tSNR_\tR\!=\!12}$),
the iteration numbers are similar around 20$\pm$1. Thus we focus on the ${\tSNR_\tT\!=\tSNR_\tR\!=\!21}$ results. In our problem setting, where only transmit beamforming vectors are optimized, increasing $M$ and decreasing $R$, decreases the iteration number. Lower $R$ yields lower SINR, in other words, the range of numbers in consideration is lower, thus the iteration number is lower. Although higher $M$ increases SINR, the iteration number decreases because, again, providing more resources, i.e., more antennas, to the transmitters returns better results, i.e., less number of iterations. The algorithm converges faster because the degree of freedom is richer due to more resources, variables, available for the optimization of transmit beamforming vectors.

\begin{figure}[!t]
\centering
\begin{subfigure}[t]{0.5\textwidth}
\centering
  \includegraphics[height=1.8in, width=3.55in] {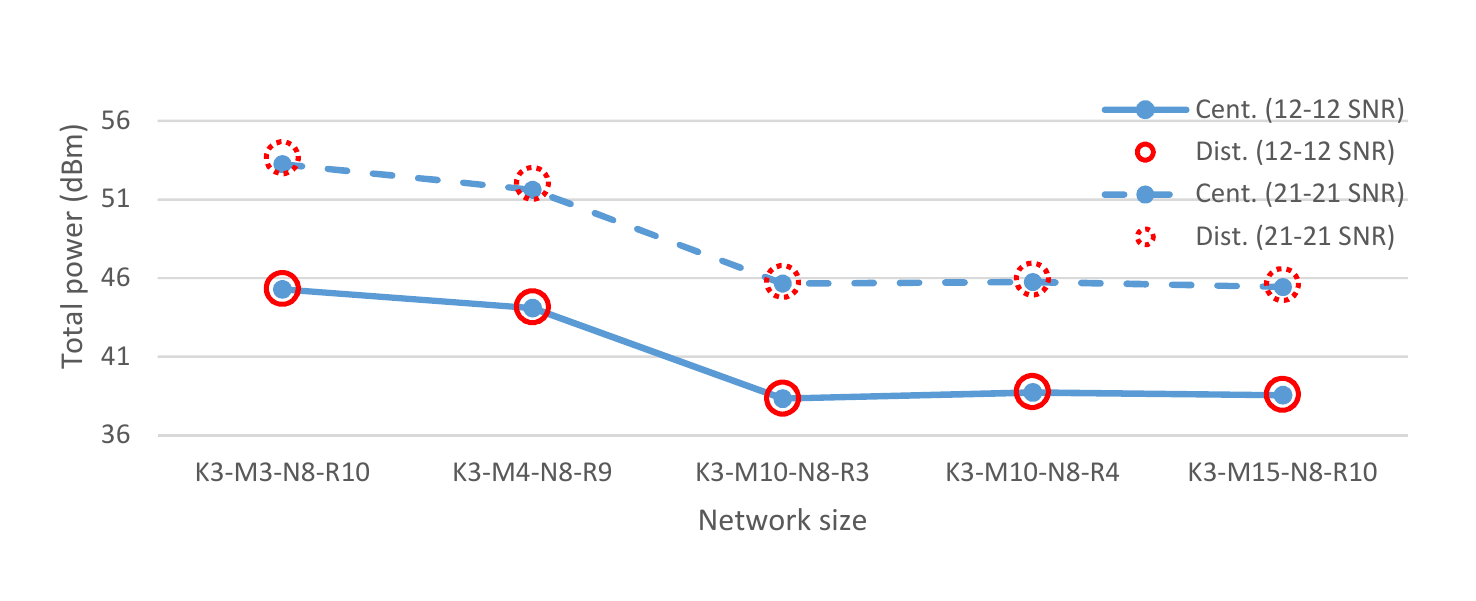}
  \caption{Average total power consumption vs. network sizes and SNRs.}
  \label{Fig:TPvsNS}
\end{subfigure}

\begin{subfigure}[t]{0.5\textwidth}
\centering
  \includegraphics[height=1.8in, width=3.55in] {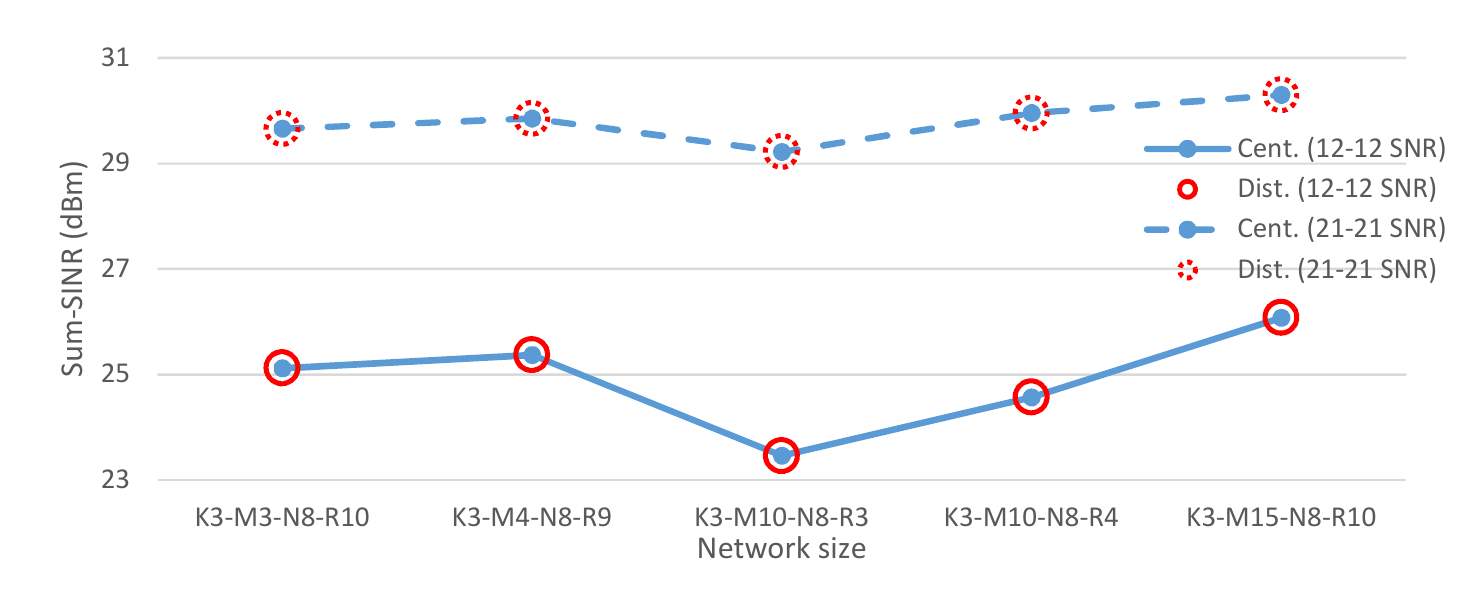}
  \caption{Average sum-SINR vs. network sizes and SNRs.}
  \label{Fig:SSvsNS}
\end{subfigure}

\begin{subfigure}[t]{0.5\textwidth}
 \centering
 \includegraphics[height=1.8in, width=3.55in] {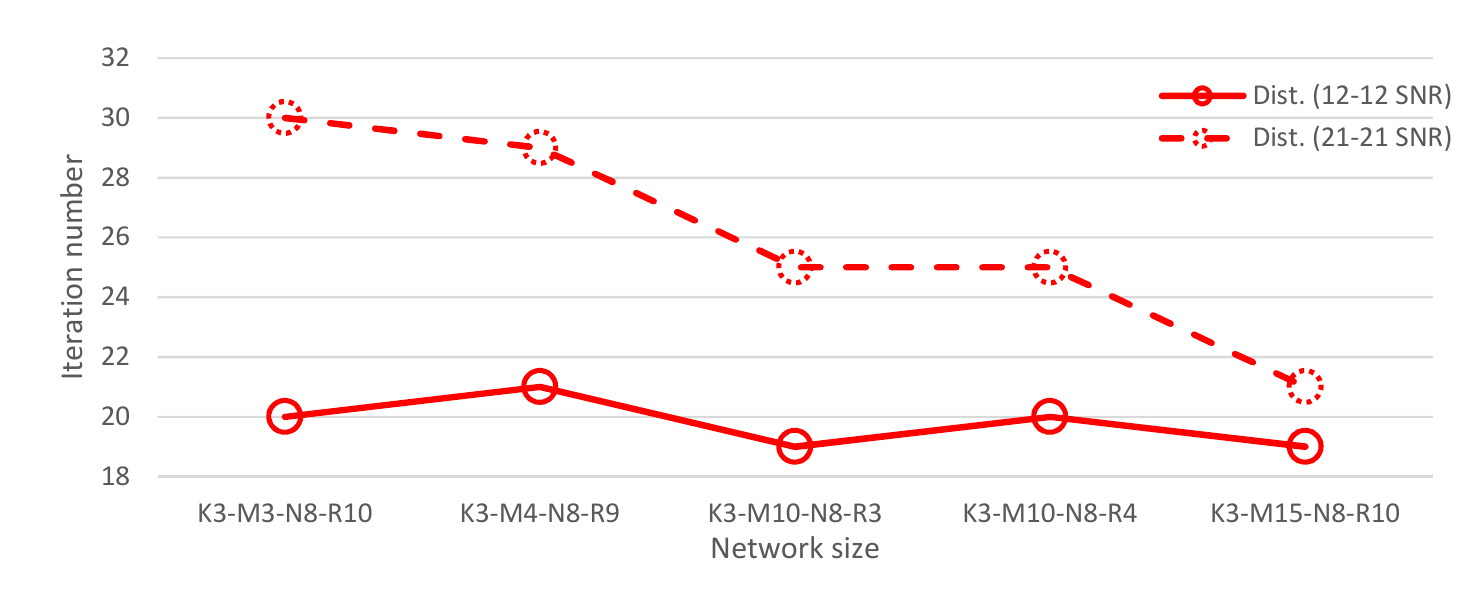}
  \caption{Average total number of iterations of the proposed distributed algorithm.}
  \label{Fig:INvsNS_12122121}
\end{subfigure}
\caption{Average total power consumptions, sum-SINRs, and total number of iterations vs. network sizes and SNRs. }
\label{Fig:TPSSINvsNS}
 \end{figure}

\subsubsection{Performances over a Wider Range of SNRs for a Particular Network Size (Fig. \ref{Fig:TPSSINvsSNR})}
In Figs. \ref{Fig:TPvsSNR}, \ref{Fig:SSvsSNR}, and \ref{Fig:INvsSNR}, the numerical results of average total power, sum-SINR and iteration number vs. transmit and relay SNRs for the $\mXnVIIIrIII$ network are presented, respectively. Again, the proposed algorithm achieves the optimal centralized solutions as seen in Fig. \ref{Fig:TPvsSNR}. In general, increasing SNR increases the power consumption and also sum-SINR. However, consider the blue trend lines in Figs. \ref{Fig:TPvsSNR} and \ref{Fig:SSvsSNR}. When $\tSNR_\tT$ is sufficiently large, e.g., $\tSNR_\tT\!=\!42$, increasing $\tSNR_\tR$ can have a small effect on sum-SINR as seen in the $42\tm12$, $42\tm21$, and $42\tm42$ dB results of Fig. \ref{Fig:SSvsSNR}. Next, consider the green trend lines in Figs. \ref{Fig:TPvsSNR} and \ref{Fig:SSvsSNR}. When $\tSNR_\tR$ is sufficiently large, increasing $\tSNR_\tT$ is advantageous as seen in $12\tm42$, $21\tm42$, and $42\tm42$ dB results of Figs. \ref{Fig:TPvsSNR} and \ref{Fig:SSvsSNR}. $\tSNR_\tR=42$ is sufficiently large to assist in achieving higher SINR targets, and providing more resources to the transmitter side returns better results, i.e., higher SINRs are achieved with lower power consumptions. On the other hand, $\tSNR_\tR=12 \text{ and } 21$ values are not sufficiently large. Thus, both the total power consumption and sum-SINR are increased by the increasing $\tSNR_\tT$, e.g., consider the $12\tm21$, $21\tm21$, and $42\tm21$ dB results.

In Figs. \ref{Fig:TPvsSNR} and \ref{Fig:SSvsSNR}, the numerical results for another benchmark that assumes the direct links between transmitters and receivers do not exist are presented. When $\tSNR_\tT$ is high, more power is consumed while less sum-SINR is achieved as seen in the $42\tm12$ and $42\tm21$ dB results. However, when $\tSNR_\tR$ is high, the effect of nonexisting direct links vanishes as seen in the $21\tm42$ and $42\tm42$ dB results.

As seen in Fig. \ref{Fig:INvsSNR}, in general, the iteration number increases as SNR increases, by increasing both $\tSNR_\tT$ and $\tSNR_\tR$, e.g., $12\tm12$, $21\tm21$, and $42\tm42$ dB, and by only increasing $\tSNR_\tT$, e.g., $12\tm12$, $21\tm12$, and $42\tm12$ dB. However, when only $\tSNR_\tR$ is increased and reaches up to $42$ dB, the iteration number decreases, i.e., compare $12\tm12$ and $12\tm21$ with $12\tm42$, and compare $21\tm12$ and $21\tm21$ with $21\tm42$ dB results. As seen in Fig. \ref{Fig:TPvsSNR}, there are notable peaks in power consumptions at the $12\tm42$ and $21\tm42$ dB points. The disproportionately high $\tSNR_\tR$ value with respect to $\tSNR_\tT$ helps in rapidly achieving the SINR targets, which cannot be controlled by the transmitter side, before the transmit beamforming optimization can further reduce the total power consumption. When the marginal ${\tSNR_\tT=42}$ results are compared, as explained earlier, similar sum-SINR results are achieved for these dB points while total power consumptions are reduced from left to right as displayed in Fig. \ref{Fig:TPvsSNR}. As mentioned earlier, smaller range of numbers of interest, i.e., smaller power consumption values, decreases the iteration number. Thus the iteration number decreases gradually in the last three dB points of Fig. \ref{Fig:INvsSNR}.

The total power savings can be obtained from the figures, i.e., total power saving (dB)${=(K p_\text{T}^{\tmax}+R p_\text{R}^{\tmax})}$ (dB)$-$(\mbox{y-axis} value) (dB). As seen in the results, e.g., Figs. \ref{Fig:TPvsNS} and \ref{Fig:TPvsSNR}, the power savings are large as explained next. Due to the random initializations of beamforming vectors, the initial power budgets are likely to be highly redundant to achieve the SINR targets determined by these beamforming vectors. On the other hand, if \mbox{max-SINR} filters are used for the initializations instead of random initializations to determine the SINR targets, the SINR targets are likely to be infeasible.

\subsubsection{Supporting Results for Disproportionate SNRs over Network Sizes (Fig. \ref{Fig:INvsNS_12424212})}
In Fig. \ref{Fig:INvsNS_12424212}, the early convergence of the algorithm when $\tSNR_\tR$ is disproportionately high is shown again over different network sizes. Since the algorithm achieves the SINR targets rapidly due to the high $\tSNR_\tR$, the transmitter side lacks the opportunity to further reduce the power consumption similar to the case observed in Fig. \ref{Fig:TPSSINvsSNR}.

 \begin{figure}[!t]
 \centering
 \begin{subfigure}[t]{0.5\textwidth}
 \centering
\hspace{-1cm} \includegraphics[height=2.1in, width=3.5in] {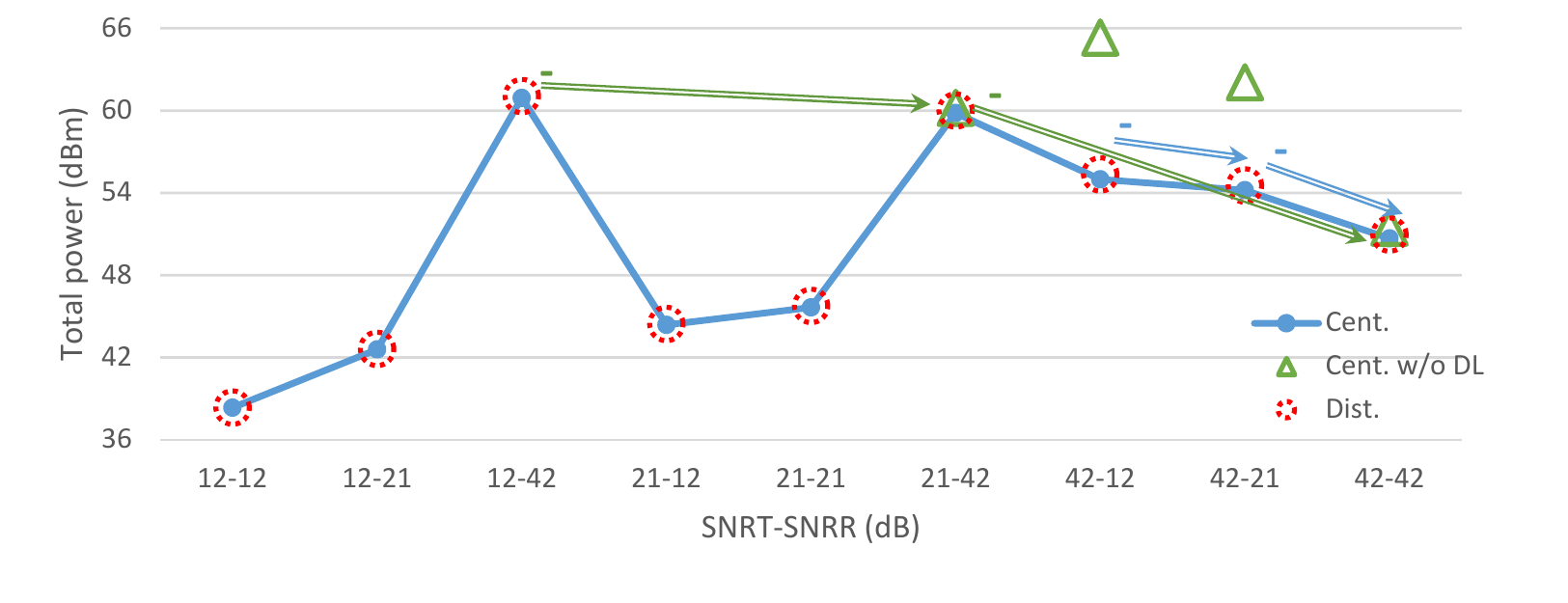}
   \caption{Average total power consumption vs. SNRs. }
  \label{Fig:TPvsSNR}
\end{subfigure}

 \begin{subfigure}[t]{0.5\textwidth}
 \centering
\hspace{-1cm} \includegraphics[height=2in, width=3.5in] {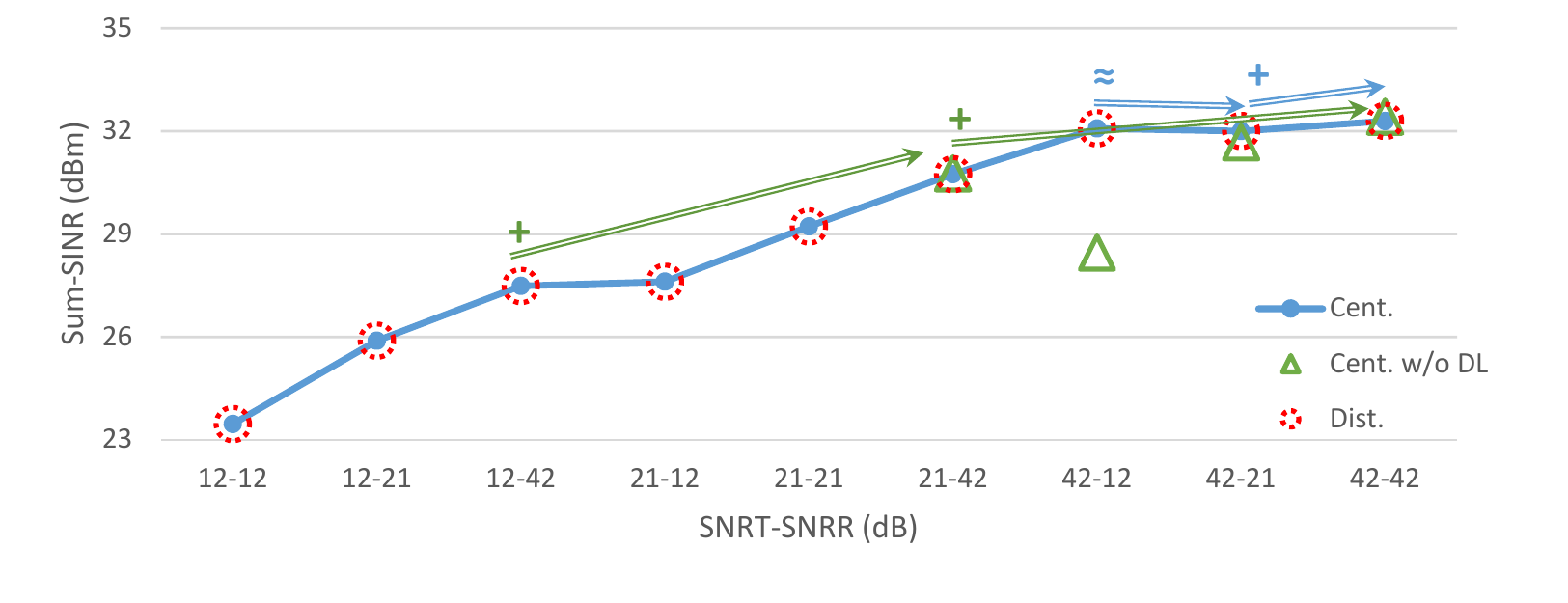}
  \caption{Average sum-SINR vs. SNRs. }
  \label{Fig:SSvsSNR}
 \end{subfigure}

\begin{subfigure}[t]{0.5\textwidth}
 \centering
 \includegraphics[height=1.8in, width=3.5in] {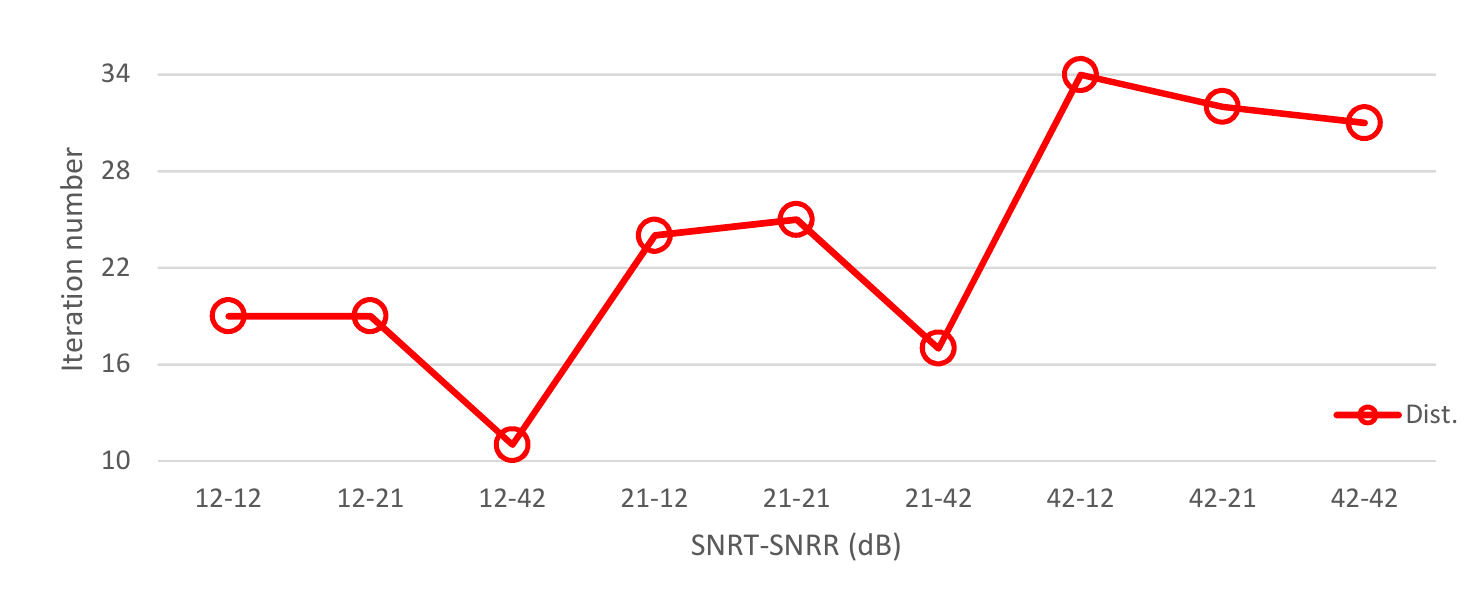}
  \caption{Average total number of iterations vs. SNRs of the proposed distributed algorithm.}
  \label{Fig:INvsSNR}
\end{subfigure}

  \caption{Average total power consumption, sum-SINR, and total number of iterations vs. SNRs in the $\mXnVIIIrIII$ network configuration.
Green lines with arrows indicate the trends.}
  \label{Fig:TPSSINvsSNR}
\end{figure}

\begin{figure}[!t]
 \centering
 \includegraphics[height=1.8in, width=3.55in] {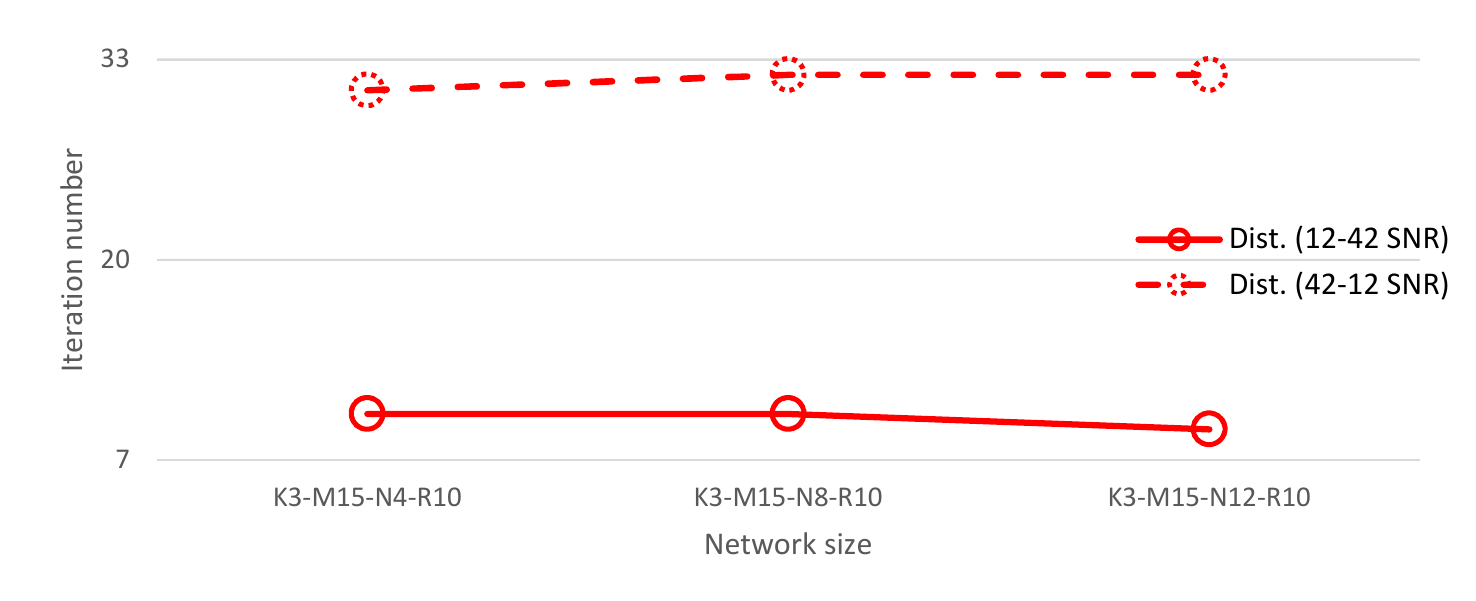}
  \caption{Average total number of iterations vs. network sizes and SNRs of the proposed distributed algorithm. }
  \label{Fig:INvsNS_12424212}
\end{figure}
  \begin{figure}[!t]
 \centering
 \includegraphics[height=1.5in, width=3.6in] {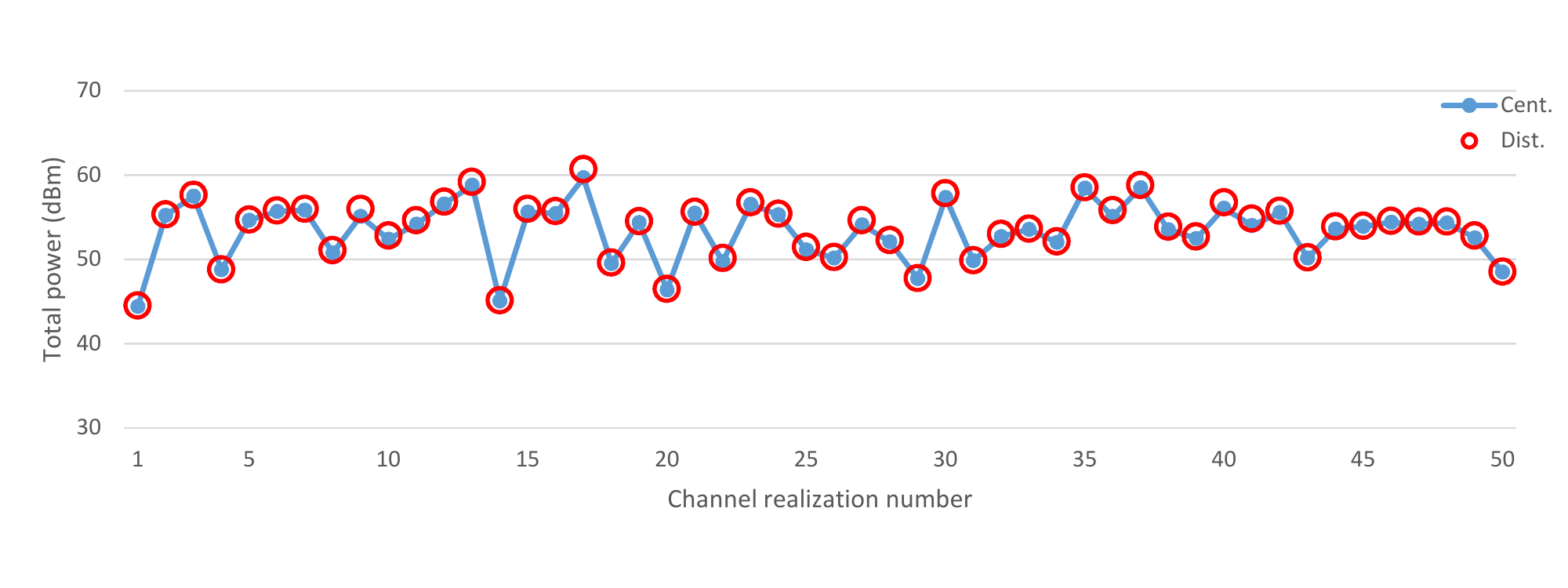}
  \caption{Total power consumption vs. channel realizations of the centralized and proposed distributed algorithms in the $\mXnVIIIrIII$ network configuration at $\tSNR_\tT=42$ and $\tSNR_\tR=12$ dB.}
  \label{Fig:TPvsCRN}
\end{figure}
 \begin{figure}[!t]
 \centering
  \begin{subfigure}[h]{0.5\textwidth}
 \centering
 \includegraphics[height=1.5in, width=3.65in] {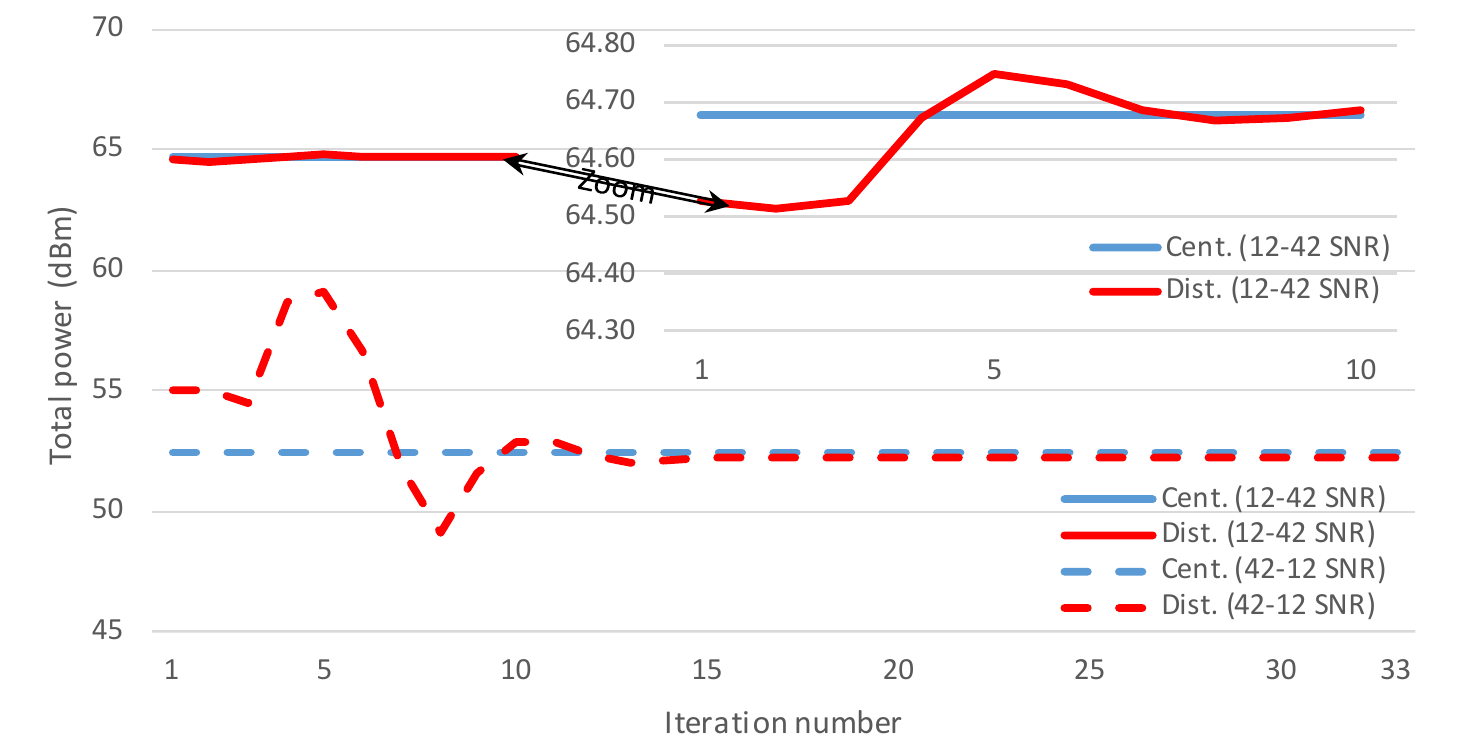}
 \caption{Total power consumption vs. iteration number.}
  \label{Fig:TPvsIN}
\end{subfigure}

 \begin{subfigure}[h]{0.5\textwidth}
 \centering
 \includegraphics[height=1.5in, width=3.5in] {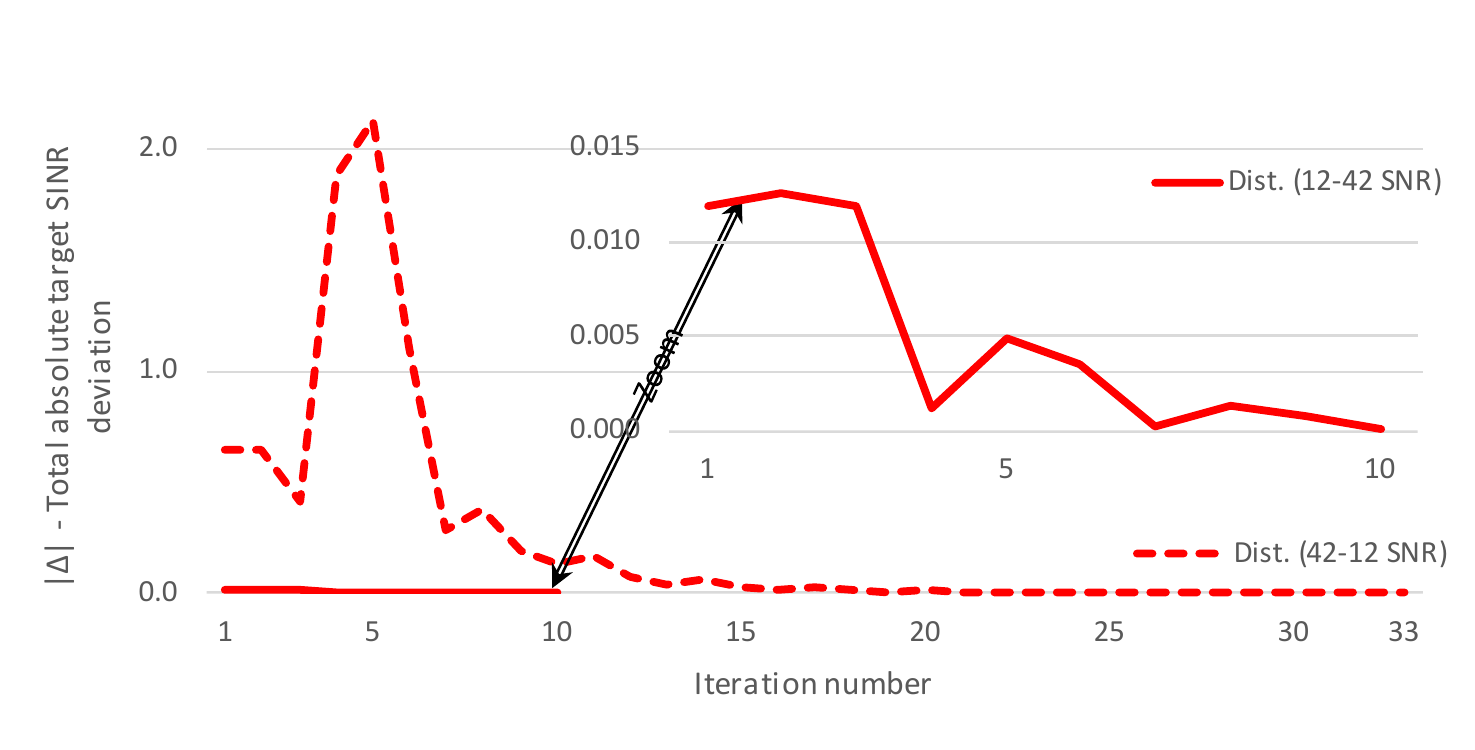}
  \caption{Total absolute target SINR deviation vs. iteration number.}
  \label{Fig:DeltavsIN}
\end{subfigure}

\caption{Total power consumption and absolute target SINR deviation vs. iteration and SNRs in the $\mXVnVIIIrX$ network configuration. }
\label{Fig:TPDeltavsIN}
\end{figure}

\subsubsection{Scrutinized Performances over Multiple and Single Channel Realizations (Figs. \ref{Fig:TPvsCRN} and \ref{Fig:TPDeltavsIN})}
In Fig. \ref{Fig:TPvsCRN}, total power consumptions vs. channel realizations are presented  to demonstrate that the proposed algorithm achieves the optimal centralized solutions at each channel realization. In Figs. \ref{Fig:TPvsIN} and \ref{Fig:DeltavsIN}, the total power consumptions and total absolute target SINR deviations, $|\Delta|\triangleq|\sum_{k,l}^{K,d_k}\Delta\xkl|$ vs. iteration numbers are presented for a randomly selected channel realization.

\subsection{Proposed ADMM vs. Existing Distributed Algorithms (Table \ref{tab:180329_1450}, Figs. \ref{Fig:ADMMBGvsProposed} and \ref{Fig:ADALvsProposed_IN})}\label{subsec:vsExisting}

The average CPU times for the aforementioned distributed algorithms in the $\mVInVIIIrXIV$ network at ${\tSNR_\tT\!=\!\tSNR_\tR\!=\!21}$ dB are shown in Table \ref{tab:180329_1450}, based on a desktop computer with 64-bit operating system, Intel i7, CPU 3.40 GHz, and 16 GB RAM. A similar comparative trend is exhibited for other network configurations in our extensive experiments. The numerical results in Table \ref{tab:180329_1450} match with the analytical results obtained in Section \ref{subsec:CompComplexity}.

As seen in Fig. \ref{Fig:ADMMBGvsProposed}, the iteration number of the proposed algorithm is significantly lower than ADMM-BG.  The iteration numbers of the two algorithms are more distinct when disproportionate resources are allocated at the transmitter side, i.e., higher $M$ and $\tSNR_\tT$.
As expected, both algorithms achieve the SINR targets and the optimal centralized solutions, which are not plotted to avoid duplicate results.

In Fig. \ref{Fig:ADALvsProposed_IN}, the  iteration numbers of the proposed and ADAL algorithms are compared, respectively. Each network configuration is tested at three different SNRs,  particularly, from left to right at $12\tm12$, $21\tm21$, and $42\tm42$ dBs. For the simplicity of figures, SNRs are not noted. As seen in Fig. \ref{Fig:ADALvsProposed_IN}, due to the significantly higher number of constraints in ADAL than the proposed algorithm, the iteration number of our proposed algorithm is always significantly lower than ADAL.
\begin{table}[!t]
  \begin{center}
  \caption{AVERAGE CPU TIMES OF THE PROPOSED AND DISTRIBUTED ALGORITHMS FOR THE $\mVInVIIIrXIV$ NETWORK.}\label{tab:180329_1450}\vspace{-.2cm}
    \begin{tabular}{|c|c|c|c|}
      \hline
      Algorithm & Proposed & ADMM-BG & ADAL \\ \hline
      Average CPU time (sec.) & 1.48 & 1.98 & 2.13 \\
      \hline
    \end{tabular}
  \end{center}
\end{table}

  \begin{figure}[!t]
 \centering
 \includegraphics[height=1.5in, width=3.5in] {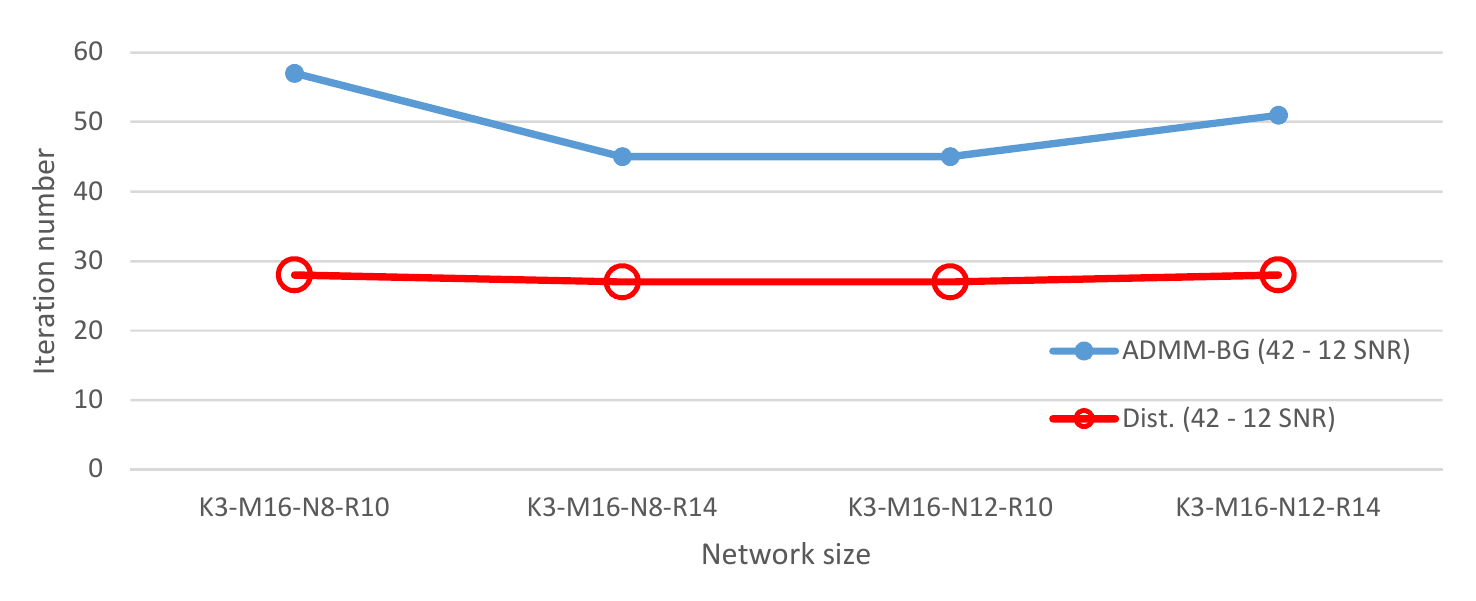}
  \caption{Average total number of iterations vs. network sizes.}
  \label{Fig:ADMMBGvsProposed}
\end{figure}

  \begin{figure}[!th]
 \centering
 \includegraphics[height=1.5in, width=3.5in] {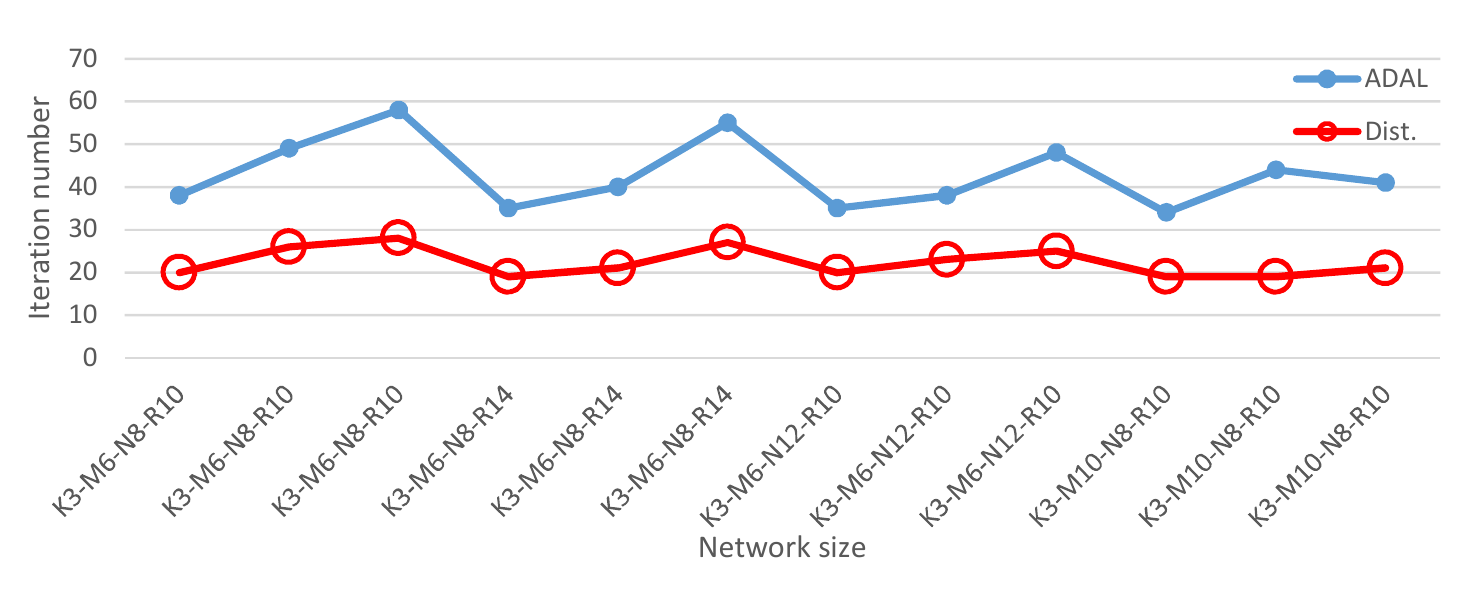}
\caption{Average total number of iterations vs. network sizes and SNRs. Each network configuration is tested three times at three different SNRs.}
  \label{Fig:ADALvsProposed_IN}
\end{figure}

\subsection{BER Performance of Proposed ADMM (Fig. \ref{Fig:BERimprovement})}\label{subsec:BER}
The important application areas of the SINR apparatus are deriving the closed-form BER and outage expressions \cite{738,739}. Therefore, as mentioned earlier, stream SINR is interconnected with the BER metric. Contrary to the rate results, the BER results cannot be obtained from the numerical SINR results presented in this section. In general, BER and outage performances improve as SINR improves, e.g., SINR outage is the probability of SINR descending below a preset SINR target. To illustrate this general trend, the \mbox{sum-SINR} and the corresponding BER improvement percentages of the system are demonstrated in Fig. \ref{Fig:BERimprovement}. For instance, at $\tSNR_\tR=21$ dB, BER is improved, i.e.,
decreased, by $1.59\%$ by increasing the $\tSNR_\tT$ from $6$ to $12$ dB, respectively. At $\tSNR_\tR=12$ dB, $0.37\%$ BER improvement is observed by increasing the $\tSNR_\tT$ from $3$ to $6$ dB.

The results in Fig. \ref{Fig:BERimprovement} are obtained over $300$ Monte Carlo channels and uncoded QPSK modulation is used. For the earlier numerical results, the extraction of  beamforming vectors $\xu\xkl$ from $\xX\xkl$ is not needed. However, for BER simulations, $\xu\xkl$ needs to and can be extracted by rank-one decomposition \cite{728} since the ranks of $\xX\xkl$ matrices are observed to be always $1$ \cite{750,751}. Based on this observation, the proposed SDR solution is also the global optimal solution to P\ref{op:0720_1620}.

Contrarily, the maximum transmit powers $\tSNR_\tT$ and $\tSNR_\tR$ can be fixed and  higher
SINRs can still be achieved by searching for competitive SINR targets as demonstrated in the next section.

\begin{figure}[!t]
 \centering
 \includegraphics[height=1.5in, width=3.55in] {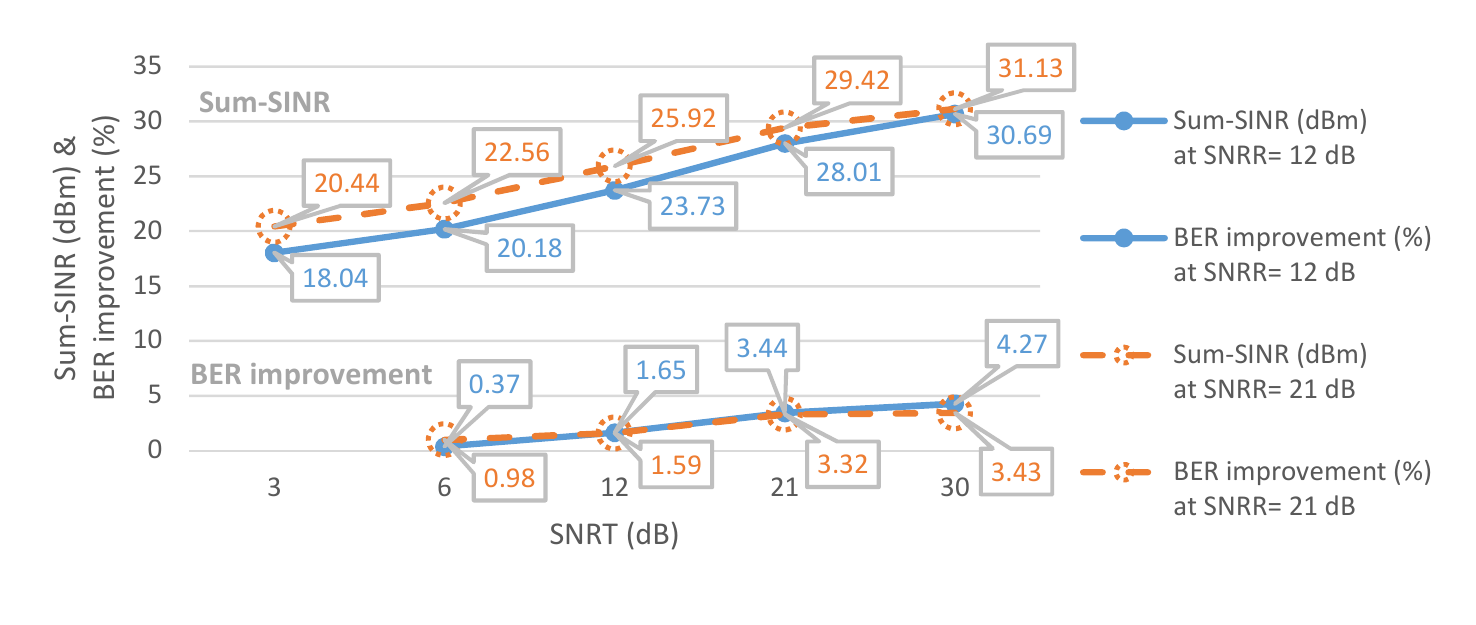}
  \caption{Average sum-SINR and the percentage of BER improvement  vs. $\tSNR_\tT$  of the proposed distributed algorithm in the $\mXnVIIIrIII$ network configuration.}
  \label{Fig:BERimprovement}
\end{figure}

{\subsection{Searching for Competitive SINR Targets (Fig. \ref{Fig:FeasibleSINRt})}\label{subsec:SINRtargets}

In this section, random filter and max-SINR filter initializations are considered to determine the lower and upper bounds of SINR targets,  respectively. In Fig. \ref{Fig:FeasibleSINRt}, the leftmost bars indicate the results of the proposed distributed algorithm with random filter initializations, which is applied in all simulations of earlier sections. The rightmost bars present the results of distributed max-SINR algorithm. The bars on the left and right sides of the middle indicate the results of the proposed distributed algorithm with $1$ iteration and $2$ iterations of linear SINR target search, respectively. As the iteration number is increased, results closer to the SINRs achieved by the \mbox{distributed max-SINR} algorithm can be obtained. In summary, by fixing the transmit power constraints, but by searching for higher feasible SINR targets, higher SINRs can be achieved at the cost of increased
power consumptions. This is an  important design trade-off in wireless networks. The results in Fig. \ref{Fig:FeasibleSINRt} are obtained over $20$ Monte Carlo channels.

\begin{figure}[!t]
 \centering
 \includegraphics[height=1.65in, width=3.55in] {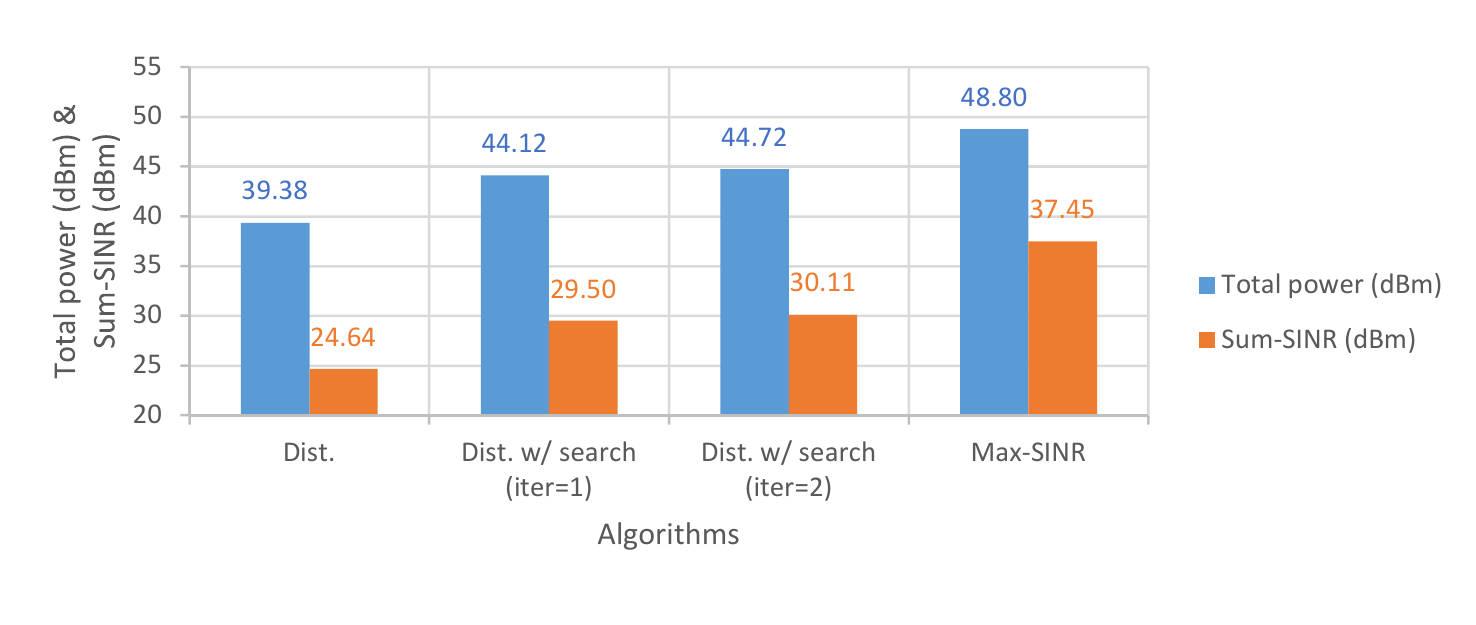}
  \caption{Average total power consumptions and sum-SINRs in the $\mVIIInVIIIrIII$ network configuration at 12-12 dB.}
  \label{Fig:FeasibleSINRt}
\end{figure}

\subsection{Distributed Joint Transmit and Relay Beamforming Filter Design (Table \ref{tab:180829_1509})}\label{subsec:DistributedJoint}
Distributed joint transmit and relay beamforming filter design improves the total power saving to achieve the same \mbox{sum-SINR} at the cost of increased complexity. Along the similar lines of the complexity analysis for transmitter side in Section \ref{subsec:CompComplexity}, the additional complexity due to the added distributed relay beamforming design is obtained as ${\mathcal{O}\big(B(2N+5)\big)}$ per relay, aka processor. Hence, as seen in \mbox{Table \ref{tab:180829_1509}} for the $\mXnXrIII$ network, the additional total complexity for all relays per iteration is given by ${\mathcal{O}\big(RB(2N+5)\big)}$ or ${\mathcal{O}\big(3\!\cdot\!6\!\cdot\!(2\!\cdot\!10\!+\!5)\!=\!450\big)}$, where ${B=6}$ since each user transmits $2$ streams as noted in the beginning of the section. On the other hand, as obtained in Section \ref{subsec:CompComplexity}, the total complexity of distributed transmitter optimization for all streams per iteration is  ${\mathcal{O}(B(M\!+\!5)\!=\!6\!\cdot\!(10\!+\!5)\!=\!90)}$.
Each relay serves all streams to achieve the  SINR targets. Thus, the distributed optimization at the relay side has higher number of constraints that substantially increases the total complexity. For the special case ${M\!=\!N\!\gg\!5}$, the ratio of the total complexities of the relay and transmitter sides is obtained as
\begin{equation}
\mathcal{O}\Bigg(\frac{RB(2N+5)}{B(M+5)}\Bigg)\approx\mathcal{O}(2R).
\end{equation}

Due to the joint optimization, the average iteration number also increases as seen in  Table \ref{tab:180829_1509}. Therefore, the average total complexities of the distributed transmit
and the distributed joint transmit and relay beamforming design algorithms are given as ${\mathcal{O}(16\cdot90=1440)}$ and ${\mathcal{O}(23\cdot540=12420)}$, respectively. The first three rows of results in Table \ref{tab:180829_1509} are obtained via numerical results averaged over $20$ Monte Carlo channels. The total power saving increases as more resources are allocated to the relay side, i.e., ${R>K}$, ${N>M}$, and ${\tSNR_\tR>\tSNR_\tT}$.

The conclusions drawn from the numerical results of distributed transmit beamforming optimization  are also valid for distributed joint transmit and relay beamforming optimization. For instance, recall the total power and sum-SINR results of increasing $\tSNR_\tT$ ${12\tm42}$, ${21\tm42}$, ${42\tm42}$ dB and increasing $\tSNR_\tR$ ${12\tm12}$, ${12\tm21}$, ${12\tm42}$ dB points in Fig. \ref{Fig:TPSSINvsSNR}. For the former case, higher sum-SINRs are achieved with lower power consumptions. On the other hand, for the latter case, increasing $\tSNR_\tR$ simply increases both the total power consumption and the sum-SINR.  Clearly, for distributed transmit beamforming optimization, providing more resources to the transmitter side returns better results. On the other hand, for distributed joint transmit and relay beamforming optimization, providing more resources to the dominant side returns better results. For instance, when the relay side is dominant ${R\gg K}$, clearly, increasing $\tSNR_\tR$ returns better results than increasing $\tSNR_\tT$.

\begin{table}
  \caption{PERFORMANCES OF DISTRIBUTED TRANSMIT VS. DISTRIBUTED JOINT TRANSMIT AND RELAY BEAMFORMING FILTER DESIGNS FOR THE $\mXnXrIII$ NETWORK AT $21\tm21$ dB.}\label{tab:180829_1509}\vspace{-.2cm}
  \begin{center}
    \begin{tabular}{|c|c|c|}
      \hline
      \multirow{2}{*}{\backslashbox{Metric}{Algorithm}} & Dist.  & Dist. joint\\
      & $\{\xU_k\}$ & $\{\xU_k,\xF_r\}$ \\ \hline
      Average sum-SINR (dBm) & 28.32 & 28.32 \\ \hline
      Average total power (dBm) & 46.05 & 45.23 \\ \hline
      Average iteration number & 16 & 23 \\ \hline 
      \!\!\!\!Total complexity per iteration\!\!\!\! &\! $\mathcal{O}(90)$  & \!$\mathcal{O}(90\!+\!450\!=\!540)$ \\ \hline
      Average total complexity &\!\!\!\!\! ${\mathcal{O}(16\!\cdot\!90\!=\!1440)}$ \!\!\! &\!\!\!\!\! ${\mathcal{O}(23\!\cdot\!540\!=\!12420)}$\!\!\! \\
      \hline
    \end{tabular}
  \end{center}
\end{table}

\section{Conclusion}\label{sec:Conclusion}
A distributed ADMM algorithm is proposed to design transmit beamforming matrices for a generic wireless relay network that has been hardly studied in the literature due to the challenges raised by the coexistence of multi-stream transmissions, multiple multi-antenna nodes, and the presence of direct links. The traits of the proposed algorithm are  low complexity,  iteration number,  and message exchange loads. Due to the challenge that the direct links bring, an approximate SINR formulation at the relay side is proposed to design distributed joint transmit and relay beamforming filters that further improve the total power saving at the cost of  increased complexity.
\begin{appendices}
\section{Closed-Form Solutions of MAX-SINR Filters}\label{app:DIA}
Both downlink and uplink iterations of the distributed \mbox{max-SINR} algorithm are based on the covariance matrix of interference signals
plus noise, which is second-order statistics information. The covariance matrix of downlink direction is similar to the one given for interference channels, i.e., there are no
relay nodes, in \cite{55}. In downlink direction, based on the covariance matrix at stream $(k,l)$
\begin{align}
\xQ\xkl=&\!\!\!\sum_{\substack{j=1,j\neq k}}^{K}\!\!\!\xH\xxkj\xU_j\xU_j^H\xH\xxkj^H+\!\!\!\!\!\!\sum_{\substack{m=1,m\neq l}}^{d_k}\!\!\!\!\!\!\xH\xxkk\xu\xkm\xu\xkm^H\xH\xxkk^H\nonumber\\
&+\xR_{n_k}/d_{k},
\end{align}
the receive filter is obtained as
\begin{equation}\label{eqn:rxfilter1}
\xv\xkl=\frac{\xQ\xkl^{-1}\xH\xxkk\xu\xkl}{||\xQ\xkl||}.
\end{equation}

For the uplink iteration,  we initially need to obtain the reciprocal channel $\overleftarrow{\xH}\xkRi^\prime$ of $\xH\xkRi^\prime$ in \eqref{eqn:effectivechannel}. The effective channel in
the uplink direction from transmitter $i$ to receiver $k$ through all relays is given as
\begin{equation}
 \overleftarrow{\xH}^\prime\xkRi=\sum_{r=1}^R\overleftarrow{\xH}_{kr}^{\prime\prime}\overleftarrow{\xF}_r\overleftarrow{\xG}_{ri},
\end{equation}
where $\overleftarrow{\xG}_{ri}\triangleq{\xG}_{ir}^H$ and $\overleftarrow{\xH}_{kr}^{\prime\prime}\triangleq\xH_{rk}^H$ are the ${N_r\times M_i}$ and ${M_k \times N_r}$ channel matrices
between transmitter $i$ and relay $r$, and relay $r$ and receiver $k$, respectively.
The covariance matrix of uplink direction can be obtained by interpreting the time domain, i.e., the two time slots due to the relay network communication, as the space domain.  Thus, the covariance matrix at stream $(k,l)$ is given as
\begin{align}
\overleftarrow{\xQ}\xkl=&\sum_{\substack{j=1,j\neq k}}^{K}\Big(\overleftarrow{\xJ}\xxkj\overleftarrow{\xU}_j(1)\overleftarrow{\xU}_j^H(1)\overleftarrow{\xJ}\xxkj^H\nonumber\\
&+\overleftarrow{\xH}^\prime\xkRj\overleftarrow{\xU}_j(2)\overleftarrow{\xU}_j^H(2)\overleftarrow{\xH}^{\prime H}\xkRj \Big)\nonumber\\
&+\!\!\!\sum_{\substack{m=1,m\neq l}}^{d_k}\Big(\overleftarrow{\xJ}\xxkk\overleftarrow{\xu}\xkm(1)\overleftarrow{\xu}\xkm^H(1)\overleftarrow{\xJ}\xxkk^H\nonumber\\
&+\overleftarrow{\xH}^\prime\xkRk\overleftarrow{\xu}\xkm(2)\overleftarrow{\xu}\xkm^H(2)\overleftarrow{\xH}^{\prime H}\xkRk \Big)+\overleftarrow{\xR}_{n_k}/d_k,
\end{align}
where $\overleftarrow{\xJ}\xxkj\triangleq\xJ_{jk}^H$ is the $M_k\times M_j$ channel matrix for the direct link between transmitter $j$ and receiver $k$, and
\begin{equation}
\overleftarrow{\xR}_{n_k}\!\triangleq\left(\sigma_k^2(1)+\sigma_k^2(2)\right)\xI_{M_k}\!\!+\!\sigma_r^2\sum_{r=1}^R\overleftarrow{\xH}_{kr}^{\prime\prime}\overleftarrow{\xF_r}\overleftarrow{\xF}_r^H\overleftarrow{\xH}_{kr}^{\prime\prime
H}
\end{equation}
is the noise covariance matrix. Then,
\begin{equation}
\overleftarrow{\xv}\xkl^\prime=\frac{\overleftarrow{\xQ}\xkl^{-1}\left(\overleftarrow{\xJ}\xxkk\overleftarrow{\xu}\xkl(1)+\overleftarrow{\xH}^\prime\xkRk\overleftarrow{\xu}\xkl(2)\right)}{||\overleftarrow{\xQ}\xkl||}
\end{equation}
and finally, the receive filter is obtained as
\begin{equation}\label{eqn:rxfilter2}
\overleftarrow{\xv}\xkl=\sqrt{p_k/d_k}\frac{\overleftarrow{\xv}\xkl^\prime}{||\overleftarrow{\xv}\xkl^\prime||}.
\end{equation}

As detailed in \cite{55}, after obtaining the receive filter \eqref{eqn:rxfilter1} in the downlink direction, the uplink direction is iterated by setting the transmit filter as ${\overleftarrow{\xu}\xkl=\xv\xkl}$.
After obtaining the receive filter \eqref{eqn:rxfilter2} in the uplink direction, the downlink direction is re-iterated by setting the transmit filter as ${\xu\xkl=\overleftarrow{\xv}\xkl}$.
The uplink and downlink iterations are continued until the convergence.

\section{Derivation Details of SINR Approximation \eqref{eqn:171212_1640}} \label{app:SINRApprox}
The term ${\zeta\xkl\yin(2)}$ in  \eqref{eqn:171124_1102} can be rewritten as
\begin{align}\label{eqn:180515_1310}
&\xv\xkl^H(2)\xH^{\prime}\xkRi\xu\xin=\xv\xkl^H(2)\sum_{r=1}^R\xG\xxkr\xF_r\xH_{ri}^{\prime\prime}\xu\xin\nonumber\\
&~~~~~~~~=\sum_{r=1}^R\xa\xxklr^H\xF_r\xb\xxrin=\xf_r^T\xc\xklrin+r_{klsin},
\end{align}
where
 \begin{align*}
 \xa\xxklr&\triangleq\xG\xxkr^H\xv\xkl(2),\xb\xxrin\triangleq\xH_{ri}^{\prime\prime}\xu\xin,\\
 r_{klsin}&\triangleq\sum_{s=1,s\neq r}^R\xa_{kls}^H\xF_s\xb_{sin},\\
 \xc\xklrin&\triangleq\xb\xxrin\otimes\xa\xxklr^* \text{, and }\xf_r\triangleq\tvec(\xF_r).
 \end{align*}
$\otimes$ and $\tvec(.)$ denotes the Kronecker product and denotes staking the columns of a matrix in a column vector, respectively. Therefore,
\begin{align}
&\big|\xv\xkl^H(2)\xH^{\prime}\xkRi\xu\xin\big|^2\!=\!(\xf_r^T\xc\xklrin+r_{klsin})(\xf_r^T\xc\xklrin\!+\!r_{klsin})^*\nonumber\\
&=\!\xf_r^T\xc\xklrin\xf_r^H\xc\xklrin^*\!+\!\xf_r^Tr_{klsin}^*\xc\xklrin\!+\!\xf_r^Hr_{klsin}\xc\xklrin^*\!+\!|r_{klsin}|^2\nonumber\\
&=\!\xf_r^H\xc\xklrin^*\xc\xklrin^T\xf_r\!+\!\xf_r^Hr_{klsin}\xc\xklrin^*\!+\!r_{klsin}^*\xc\xklrin^T\xf_r\!+\!|r_{klsin}|^2\nonumber\\
&=\!\xf_r^H\xC\xklrin\xf_r\!+\!\xf_r^H\xd\xklrin\!+\!\xd\xklrin^H\xf_r\!+\!|r_{klsin}|^2,
\end{align}
where $$\xC\xklrin\triangleq\xc\xklrin^*\xc\xklrin^T \text{ and } \xd\xklrin\triangleq r_{klsin}\xc\xklrin^*.$$
Now consider the following term in \eqref{eqn:180515_1330}
\begin{align}
&\xv\xkl^H(2)\sum_{r=1}^R\xG\xxkr\xF_r=\sum_{r=1}^R\xa\xxklr^H\xF_r\nonumber\\
&~~~~~~~~~~=[\xf_r^T\xo_{klr1}~\xf_r^T\xo_{klr2} \ldots \xf_r^T\xo_{klrN}]+\xt_{kls}^H,
\end{align}
where $$\xo_{klrn}\triangleq\xe_n\otimes\xa\xxklr^*, \xt_{kls}\triangleq\sum_{s=1\\s\neq r}^R\xF_s^H\xa_{kls}\text{, and}$$ $\xe_n$ is an $N\times1$ zero vector except 1 at the $n\yth$ row. Therefore,
\begin{align}
\big|\xv\xkl^H(2)\sum_{r=1}^R\xG\xxkr\xF_r\big|^2=&\sum_{n=1}^N\big|\xf_r^T\xo_{klrn}\big|^2+t_{kls}\nonumber\\
=&\xf_r^H\Big(\sum_{n=1}^N\xo_{klrn}^*\xo_{klrn}^T\Big)\xf_r+t_{kls},
\end{align}
where ${t_{kls}\triangleq\xt_{kls}^H\xt_{kls}}$. Hence, the noise power in the second time slot is
\begin{align}
\sigma_{n\xkl}^2(2)=&\sigma_r^2\big|\xv\xkl^H(2)\sum_{r=1}^R\xG\xxkr\xF_r\big|^2+\sigma_k^2(2)\xv\xkl^H(2)\xv\xkl(2)\nonumber\\
=&\xf_r^H\xO\xxklr\xf_r+\sigma\xxklr^2,
\end{align}

where
\begin{align*}
\xO\xxklr&\triangleq\sigma_r^2\sum_{n=1}^N\xo_{klrn}^*\xo_{klrn}^T \text{ and }\\
\sigma\xxklr^2&\triangleq\sigma_r^2t_{kls}+\sigma_k^2(2)\xv\xkl^H(2)\xv\xkl(2).
\end{align*}
Hence, the approximate SINR as a function of relay filter $\xf_r$ is given as

\begin{equation}
\widehat{\tSINR}\xkl(\xf_r)=\frac{x}{\sum_{(j,m)\neq(k,l)}\big(y\big)+\tilde{\sigma}_{n\xkl}^2},
\end{equation}
where
\begin{align*}
x\!&\triangleq\!\zeta\xkl\ykl(1)\!+\!\xf_r^H\xC_{klrkl}\xf_r\!+\!\xf_r^H\xd\xklrkl\!+\!\xd\xklrkl^H\xf_r\!+\!|r_{klskl}|^2,\\
y\!&\triangleq\!\zeta\xkl\yjm(1)\!+\!\xf_r^H\xC_{klrjm}\xf_r\!+\!\xf_r^H\xd\xklrjm\!+\!\xd\xklrjm^H\xf_r\!\!+\!\!|r_{klsjm}\!|^2,\\
\tilde{\sigma}_{n\xkl}^2\!\!&\triangleq\sigma_k^2(1)\xv\xkl^H(1)\xv\xkl(1)+\xf_r^H\xO\xxklr\xf_r+\sigma\xxklr^2.
\end{align*}

Furthermore,
\begin{equation}
\widehat{\tSINR}\xkl(\xf_r)\!\!=\!\!\frac{\xf_r^H\xC_{klrkl}\xf_r\!+\!\xf_r^H\xd\xklrkl\!+\!\xd\xklrkl^H\xf_r\!+\!r_{1,klrkl}}{\xf_r^H\bar{\xC}_{klrjm}\xf_r\!+\!\xf_r^H\bar{\xd}\xklrjm\!+\!\bar{\xd}\xklrjm^H\xf_r\!+\!r_{2,klrjm}},
\end{equation}
where
\begin{align*}
r_{1,klrkl}&\triangleq\zeta\xkl\ykl(1)+|r_{klskl}|^2,\\
\bar{\xC}_{klrjm}\!&\triangleq\!\!\!\!\!\!\!\!\sum_{(j,m)\neq(k,l)}\!\!\!\!\!\xC_{klrjm}\!+\!\xO\xxklr,\bar{\xd}\xklrjm\triangleq\!\!\!\!\!\!\!\!\!\sum_{(j,m)\neq(k,l)}\!\!\!\!\!\xd\xklrjm \text{, and }\\
r_{2,klrjm}\!&\triangleq\!\!\!\!\!\!\!\sum_{(j,m)\neq(k,l)}\!\!\!\!\!\big(\zeta\xkl\yjm(1)+|r_{klsjm}|^2\big)\\
&~~~~~~~~~~+\sigma_k^2(1)\xv\xkl^H(1)\xv\xkl(1)+\sigma\xxklr^2.
\end{align*}
Note that ${\widehat{\tSINR}\xkl(\xf_r)}$ is in an inhomogeneous quadratic form. To rewrite it in homogeneous quadratic form, let
\begin{align}
\tilde{\xC}_{klrkl}\triangleq&\left[
  \begin{array}{cc}
    \xC_{klrkl} & \xd\xklrkl \\
    \xd\xklrkl^H & r_{1,klrkl} \\
  \end{array}
\right] \text{ and }\nonumber\\
\tilde{\xC}_{klrjm}\triangleq&\left[
  \begin{array}{cc}
    \bar{\xC}_{klrjm} & \bar{\xd}\xklrjm \\
    \bar{\xd}\xklrjm^H & r_{2,klrjm} \\
  \end{array}
\right].
\end{align}
Thus, \eqref{eqn:171212_1640} is obtained.
\end{appendices}

\bibliographystyle{IEEEtran}

\end{document}